\documentclass[12pt,reqno]{amsart}
\usepackage{amsmath}
\usepackage{latexsym}
\usepackage{amsfonts}
\usepackage{amssymb}
\usepackage{color}
\usepackage{bbm,dsfont}
\usepackage{graphicx}
\usepackage{MnSymbol}
\usepackage[normalem]{ulem}
\usepackage{enumerate}

\newtheorem{proposition}{Proposition}
\newtheorem{theorem}{Theorem}
\newtheorem{lemma}{Lemma}
\newtheorem{corollary}{Corollary}

\theoremstyle{definition}

\newtheorem{example}{Example}
\newtheorem{definition}{Definition}

\newcommand{\de}{\,{\rm d}}

\newcommand{\R}{\mathbb{R}} %real
\newcommand{\C}{\mathbb{C}} %complex
\newcommand{\nat}{\mathbb N} %natural
\newcommand{\Z}{\mathbb Z} %integer
\newcommand{\half}{\tfrac{1}{2}} %half
 
\newcommand{\floor}[1]{\lfloor #1 \rfloor} %floor function
\newcommand{\bigfloor}[1]{\left\lfloor #1 \right\rfloor} %floor function

\newcommand{\pp}{\mathcal{P}}
\newcommand{\xx}{\mathcal{X}}
\newcommand{\rr}{\mathcal{R}}
\renewcommand{\tt}{\mathcal{T}}
\renewcommand{\gg}{\mathcal{G}}

\newcommand{\hi}{\mathcal{H}} %Hilbert space
\newcommand{\hh}{\mathcal{H}} %Hilbert space
\renewcommand{\ss}{\mathcal{S}} %operator system
\renewcommand{\aa}{\mathcal{A}} %von Neumann algebra
\newcommand{\lh}{\mathcal{L(H)}} %bounded linear operators
\newcommand{\lhs}{\mathcal{L}_s(\hi)} %bounded seladjoint linear operators
\newcommand{\spannoC}[1]{{\rm span}_\C \left\{ #1 \right\}}
\newcommand{\spannoR}[1]{{\rm span}_\R \left\{ #1 \right\}}
\newcommand{\trh}{\mathcal{T(H)}} %trace class operators
\newcommand{\trhs}{\mathcal{T}_s(\hi)} %selfadjoint trace class operators
\newcommand{\ip}[2]{\left\langle\,#1\,|\,#2\,\right\rangle} %inner product
\newcommand{\kb}[2]{|#1\rangle\langle#2|} %ketbra
\newcommand{\no}[1]{\left\|#1\right\|} %norm
\newcommand{\tr}[1]{{\rm tr}\left[#1\right]} %trace
\renewcommand{\det}[1]{{\rm det}\left[#1\right]} %det
\newcommand{\id}{\mathbbm{1}} %identity operator
\newcommand{\fii}{\varphi}
\renewcommand{\rho}{\varrho}
\newcommand{\lam}{\lambda}

\newcommand{\rank}{{\rm rank}\,} %rank

\newcommand{\lft}{\left(}
\newcommand{\rgt}{\right)}

\newcommand{\pair}[2]{\left\langle\,#1  , #2\,\right\rangle} %pairing
\newcommand{\ggh}{\widehat{\mathcal{G}}}

\newcommand{\bor}[1]{\mathcal{B}(#1)} % Borel sigma-algebra
\newcommand{\Co}{\mathsf{C}}%generic observable
\newcommand{\Mo}{\mathsf{M}}%generic observable
\newcommand{\prem}{\mathcal{P}}
\newcommand{\task}{\mathcal{T}} %task 
\newcommand{\state}{\mathcal{S}} %states
\begin{document}\setlength{\arraycolsep}{2pt}

\title{Tasks and premises in quantum state determination}

\begin{abstract}
The purpose of quantum tomography is to determine an unknown quantum state from measurement outcome statistics.
There are two obvious ways to generalize this setting.
First, our task need not be the determination of any possible input state but only some input states, for instance pure states.
Second, we may have some prior information, or premise, which guarantees that the input state belongs to some subset of states, for instance the set of states with rank less than half of the dimension of the Hilbert space.
We investigate state determination under these two supplemental features, concentrating on the cases where the task and the premise are statements about the rank of the unknown state.
We characterize the structure of quantum observables (POVMs) that are capable of fulfilling these type of determination tasks. 
After the general treatment we focus on the class of covariant phase space observables, thus providing physically relevant examples of observables both capable and incapable of performing these tasks. In this context, the effect of noise is discussed.
 \end{abstract}

\author[Carmeli]{Claudio Carmeli}
\address{\textbf{Claudio Carmeli}; Dipartimento di Fisica, Universit\`a di Genova, Via Dodecaneso 33, I-16146 Genova, Italy.}
 \email{claudio.carmeli@gmail.com}

\author[Heinosaari]{Teiko Heinosaari}
\address{\textbf{Teiko Heinosaari}; Turku Centre for Quantum Physics, Department of Physics and Astronomy, University of Turku, Finland}
\email{teiko.heinosaari@utu.fi}

\author[Schultz]{Jussi Schultz}
\address{\textbf{Jussi Schultz}; Dipartimento di Matematica, Politecnico di Milano, Piazza Leonardo da Vinci 32, I-20133 Milano, Italy}
\email{jussischultz@gmail.com}

\author[Toigo]{Alessandro Toigo}
\address{\textbf{Alessandro Toigo}; Dipartimento di Matematica, Politecnico di Milano, Piazza Leonardo da Vinci 32, I-20133 Milano, Italy, and I.N.F.N., Sezione di Milano, Via Celoria 16, I-20133 Milano, Italy}
\email{alessandro.toigo@polimi.it}

\maketitle

%%%%%%%%%%%%%%%%%%%%%%
\section{Introduction}\label{sec:intro}
%%%%%%%%%%%%%%%%%%%%%%

The usual task in quantum tomography is to determine an unknown quantum state from measurement outcome statistics.
There are two obvious ways to vary this setting.
First, our task need not be the determination of any possible input state but only some states belonging to a restricted subset of all states.
Second, we typically have some prior information, or premise, which tells us that the input state belongs to some subset of states.
It is clear that with this additional information and restricted task, this problem should be easier than the problem of determing an unknown quantum state without any prior information. 
As an example, consider the usual optical homodyne tomography of a single mode electromagnetic field \cite{VoRi89, Smithey93}. If the state is completely unknown, then, in principle, one needs to measure infinitely many rotated field quadratures \cite{KiSc13}. However, as soon as one knows that the state can be represented as a finite matrix in the photon number basis, then already  finitely many quadratures are enough, the exact number depending on the size of the matrix \cite{LeMu96}. It should be emphasized that the premise is not merely a mathematical assumption but carries also physical meaning. Indeed, it simply means that the probability of detecting  energies above a certain bound is zero. Since one might expect that also in general a given task and premise leads to the requirement of less or worse resources, an immediate question is the characterization of these resources.

The task and the premise can be described as subsets of the set of all states, hence this modified setting is specified by two subsets $\task$ (task) and $\prem$ (premise) of all states.
Clearly, we must have $\task\subseteq\prem$ to make the formulation meaningful. 
Smaller $\task$ means less demanding determination task and smaller $\prem$ means better prior knowledge.
In this work we study the previously explained question from the point of view of quantum observables, mathematically described as positive operator valued measures (POVMs).
A quantum observable is called \emph{informationally complete} if the measurement outcome probabilities uniquely determine each state {\cite{Prugovecki77}}, and this clearly relates to the usual task in quantum tomography.
The previously described generalized setting leads to the concept of {\em $(\task,\prem)$-informational completeness}.
We present a general formulation of this property, and then concentrate on some interesting special cases.

Our main results are related to situations when the premise tells that the rank of the input state is bounded by some number $p$, and the task is then to determine all states with rank $t\leq p$ or less. 
We show, in particular, that if there is no premise and the task is to determine all states with rank less than or equal to $\frac{d}{2}$, where $d$ is the dimension of the Hilbert space of the quantum system, then we actually need an informationally complete observable.

Perhaps the most important informationally complete observables are covariant phase space observables.
These are widely used in both finite and infinite dimensional quantum mechanics.
However, not all covariant phase space observables are informationally complete, and for instance noise can easily destroy this desired property. 
We will show that even if a covariant phase space observable fails to be informationally complete, it can be $(\task,\prem)$-informationally complete for some meaningful sets $\task$ and $\prem$.

\emph{Notation.}
We denote by $\nat$ the set of natural numbers (containing $0$) and $\nat_* = \nat\cup\{\infty\}$. 
We use the conventions $\infty \pm k = \infty + \infty =\frac{\infty}{k}=\infty$ for all nonzero $k\in\nat$. 
For every $x\in\R$, we denote by $\floor{x}$ the largest integer not greater than $x$, and we define $\floor{\infty}=\infty$.
If not specified, $\hh$ is a finite dimensional or separable infinite dimensional complex Hilbert space. 
We denote $d=\dim \hh\in\nat_*$. We denote by $\lh$ the complex Banach space of bounded linear operators on $\hh$ endowed with the uniform norm, and by $\lhs\subseteq\lh$ the real Banach subspace of selfadjoint operators. 
If $\xx \subseteq \lh$ is a complex linear space such that $A^\ast\in\xx$ whenever $A\in\xx$, we denote by $\xx_s = \xx\cap\lhs$ the selfadjoint part or $\xx$, and regard it as a real linear space. 
Then $\xx = \xx_s + i\xx_s$; in particular, $\dim_\R \xx_s = \dim_\C \xx$. We write $\trh$ for the complex Banach space of the trace class operators on $\hh$ endowed with the trace class norm, and $\trhs=\trh\cap\lhs$.
Clearly, if $\dim\hh<\infty$, then $\lh \equiv \trh$ as linear spaces.

We denote by $\state =\{\rho\in\trh\mid\rho\geq 0 \mbox{ and } \tr{\rho} = 1\}$ the set of all states (i.e., density operators) on $\hh$, and by $\state^1 = \{\rho\in\ss\mid \rho^2 = \rho\}$ the set of all pure states (i.e., one-dimensional projections).

%%%%%%%%%%%%%%%%%%%%%%
\section{Observables}
%%%%%%%%%%%%%%%%%%%%%%

In this section we generalize the linear algebra framework for quantum tomography as introduced in \cite{HeMaWo13} to a wider setting, also covering infinite dimensional Hilbert spaces and arbitrary measurable spaces.

Let $(\Omega,\aa)$ be a measurable space. An {\em observable} on $\Omega$ is a map $\Mo:\aa\to\lh$ such that
\begin{enumerate}[(1)]
\item each $\Mo(X)$ is a positive operator;
\item for all finite or denumerably infinite partitions $(X_i)_{i\in I}$ of $\Omega$ into disjoint measurable sets $X_i\in\aa$, we have $\sum_{i\in I} \Mo(X_i)=\id$, the sum converging in the weak operator topology.
\end{enumerate}
If $\Mo$ is an observable on $\Omega$ and $\rho\in\ss$, we can define the {\em associate measurement outcome probability distribution} $\rho^\Mo$ on the measurable space $(\Omega,\aa)$, given by $\rho^\Mo(X) = \tr{\rho\Mo(X)}$ for all $X\in\aa$. When $\Omega$ is a finite or denumerable set, we will take $\aa = \pp(\Omega)$, the set of all subsets of $\Omega$, and denote $\Mo(x) \equiv \Mo(\{x\})$ and $\rho^\Mo(x) \equiv \rho^\Mo(\{x\})$ for all $x\in\Omega$ for short.

A {\em weak*-closed real operator system} on $\hh$ is a weak*-closed real linear subspace $\rr\subseteq\lhs$ such that $\id\in\rr$. (Note that $\rr$ is a real operator system  if and only if $\rr_\C = {\rm span}_\C \rr$ is an operator system in the standard sense of operator theory, and then we have $\rr = (\rr_\C)_s$ {\cite{CBMOA03})}. 
If $\rr$ is a weak*-closed real operator system on $\hh$, then its {\em annihilator} is the following closed subspace of $\trhs$
$$
\rr^\perp = \{ T\in\trhs \mid \tr{TA}=0 \  \forall A\in\rr\} \, .
$$
Since $\id\in\rr$, we have $\tr{T} = 0$ for all $T\in\rr^\perp$.

Any observable $\Mo:\aa\to\lh$ generates a weak*-closed real operator system on $\hh$ as the weak*-closure of the real linear span of its range; we denote 
$$
\rr(\Mo) = \overline{\spannoR{\Mo (X) \mid X\in\aa }}^{\mathrm{w^*}} \, .
$$
Note that
\begin{align}
\label{eq:anngen}
\rr(\Mo)^\perp = \{ T\in\trhs \mid \tr{T \Mo(X)}=0 \  \forall X\in\aa \} \, .
\end{align}

Conversely, we have the following facts.

\begin{proposition}\label{prop:exM}
\begin{enumerate}[(a)]
\item Suppose $\rr$ is a weak*-closed real operator system on $\hh$. Then there exists a finite or denumerable set $\Omega$ satisfying $\#\Omega = \dim\rr$ and an observable $\Mo$ on $\Omega$ such that $\rr = \rr(\Mo)$.
\item Suppose $\mathcal{X} \subseteq\trhs$ is a closed subspace such that $\tr{T} = 0$ for all $T\in\mathcal{X} $. Then $\mathcal{X}  = \rr(\Mo)^\perp$ for some observable $\Mo$.
\end{enumerate}
\end{proposition}

\begin{proof}
\begin{enumerate}[(a)]
\item If $\dim\rr <\infty$, then this is proved in \cite[Prop.~1]{HeMaWo13} (note that the proof is not affected if $\dim\hh=\infty$). For $\dim\rr =\infty$, we use the following slight modification of the proof of \cite[Theorem~2.2]{CRQM97}. We define the set
\begin{equation*}
\rr_0 = \{A\in\rr\mid A\geq 0 \mbox{ and } \no{A}\leq 1\} \, .
\end{equation*}
Then, $\rr_0$ is weak*-compact and metrizable, being a weak*-closed subset of the unit ball of $\lh$. In particular, it is separable. Let $(A_n)_{n\in\nat}$ be a weak*-dense subset of $\rr_0$, and define an observable $\Mo : \pp(\nat_*) \to \lh$ by
\begin{equation*}
\Mo (\infty) = \id-\sum_{n=0}^\infty \frac{1}{2^{n+1}} A_n, \qquad \Mo (n) = \frac{1}{2^{n+1}} A_n \quad \mbox{for } n\geq 0 \, .
\end{equation*}
The series converges in norm, thus also in the weak*-topology.
Since $\rr_0$ is weak*-closed, we have $\Mo(\infty) \in \rr_0$. Each $A\in\rr$ can be written in the form
\begin{equation*}
A=(\no{A}\id+A)/2 - (\no{A}\id-A)/2  \, . 
\end{equation*}
Since $(\no{A}\id \pm A)/(2\no{A}) \in\rr_0$, we conclude that $\rr = {\rm span}_\R \rr_0$. 
By this fact and weak*-density of the set $(A_n)_{n\in\nat}$ in $\rr_0$, it follows that $\rr = \rr(\Mo)$.
\item Let $\rr = \{A\in\lhs\mid \tr{TA} = 0 \ \forall T\in\mathcal{X} \}$. 
It is straightforward to check that $\rr$ is a weak*-closed real operator system, hence $\rr = \rr(\Mo)$ for some observable $\Mo$ by item (a). Moreover, $\mathcal{X}  = \rr^\perp = \rr(\Mo)^\perp$ by the Bipolar Theorem; see e.g. {\cite[V.1.8]{CFA90}}.
\end{enumerate}
\end{proof}

%%%%%%%%%%%%%%%%%%%%%%
\section{($\task$,$\prem$)-informationally complete observables}
%%%%%%%%%%%%%%%%%%%%%%

We recall that an observable $\Mo$ is called \emph{informationally complete} if for any two different states $\varrho_1,\varrho_2\in\state$, the measurement outcome distributions $\varrho_1^\Mo$ and $\varrho_2^\Mo$ are different {\cite{Prugovecki77}}.
In other words, each state leads to a unique measurement outcome distribution.
In the following we generalize this concept of informational completeness.

Suppose that two nonempty subsets $\task\subseteq\prem\subseteq\state$ are given.
(They may also be equal.)
The larger subset $\prem$ corresponds to a certified premise, so we know that the initial state belongs to $\prem$ with certainty.
The smaller subset $\task$ specifies the given state determination task; we are required to determine the state whenever it belongs to $\task$.
Therefore, we can fulfill the given task if we are able to differentiate every state in $\task$ from every state in $\prem$.

This additional aspect leads to the following generalization of informational completeness.

\begin{definition}
Let $\emptyset\neq\task\subseteq\prem\subseteq\state$.
An observable $\Mo$ is \emph{$(\task,\prem)$-informationally complete} if $\varrho_1^\Mo \neq \varrho_2^\Mo$ for any two different states $\varrho_1\in\task$ and $\varrho_2\in\prem$.
\end{definition}

Clearly, an observable is informationally complete in the usual sense if and only if it is $(\state,\state)$-informationally complete.
This obviously corresponds to the most demanding state determination task without having any prior information.
We note that the smaller the set $\prem$ is, the more informative the premise is.
Likewise, the smaller the set $\task$ is, the less demanding the task is.

For all nonempty subsets $\task\subseteq\prem\subseteq\state$, we denote
\begin{equation*}
\task-\prem = \{ \varrho_1 - \varrho_2 \mid \varrho_1\in\task \, , \varrho_2\in\prem \} \subseteq \trhs
\end{equation*}
and
\begin{equation*}
\R(\task - \prem) = \{\lambda T \mid T \in \task - \prem \, ,  \lambda \in \R \} \subseteq \trhs\, .
\end{equation*}
Every $T\in \R(\task-\prem)$ satisfies $\tr{T}=0$.

\begin{proposition}\label{prop:prem}
Let $\emptyset \neq \task\subseteq\prem\subseteq\state$.
For an observable $\Mo$, the following conditions are equivalent:
\begin{enumerate}[(i)]
\item $\Mo$ is $(\task,\prem)$-informationally complete.
\item $\rr(\Mo)^\perp \cap (\task - \prem) = \{ 0 \}$.
\item $\rr(\Mo)^\perp \cap \ \R(\task - \prem) = \{ 0 \}$.
\end{enumerate}
\end{proposition}

\begin{proof}
Let $\varrho_1\in\task$ and $\varrho_2\in\prem$.
Then, {using \eqref{eq:anngen}},
\begin{align*}
\varrho_1^\Mo = \varrho_2^\Mo \quad & \Leftrightarrow \quad \forall X\in\aa : \tr{\varrho_1 \Mo(X)} = \tr{\varrho_2 \Mo(X)}  \\
& \Leftrightarrow \quad \forall X\in\aa : \tr{(\varrho_1-\varrho_2) \Mo(X)} =0   \\
& \Leftrightarrow \quad \varrho_1-\varrho_2 \in \rr(\Mo)^\perp \, .
\end{align*}
Thus, $\Mo$ is $(\task,\prem)$-informationally complete if and only if $\rr(\Mo)^\perp \cap (\task - \prem) = \{ 0 \}$. 
Since $\rr(\Mo)^\perp$ is a linear space, the equivalence of (ii) and (iii) follows.
\end{proof}

A well known mathematical characterization of informationally complete observables is that $\rr(\Mo)=\lhs$ \cite{Busch91,SiSt92}.
As an application of Prop.~\ref{prop:prem} we give a short derivation of this fact.

\begin{corollary}\label{prop:well-known}
An observable $\Mo$ is informationally complete if and only if $\rr(\Mo)^\perp = \{ 0 \}$.
\end{corollary}

\begin{proof}
Each nonzero element $T\in\rr(\Mo)^\perp$ decomposes as $T = T_+ - T_-$, where $T_\pm\geq 0$ are the positive and negative parts of $T$ \cite[p.~241]{CFA90}. Since $\tr{T} = 0$, we have $\tr{T_+} = \tr{T_-} \equiv c$, and $c \neq 0$ as otherwise $T_+ = T_- = 0$. 
Setting $\rho_{\pm} = T_\pm/c \in \state$, we  have $T = c(\rho_+ - \rho_-) \in \R(\state-\state)$. 
Thus, $\rr(\Mo)^\perp = \rr(\Mo)^\perp \cap \R(\state-\state)$, and the claim follows by Prop.~\ref{prop:prem}.
\end{proof}

%%%%%%%%%%%%%%%%%%%%%%%%%%%%%%%%%
\section{Characterization of ($\task$,$\prem$)-informational completeness in various cases}
%%%%%%%%%%%%%%%%%%%%%%%%%%%%%%%%%

The rank of an operator $A\in\lh$ is the dimension of its range: $\rank (A) = \dim (A\hh) \in\nat_*$.
For each $r\in\nat_*$ such that $1\leq r \leq d$, we denote by $\state^{\leq r}$ the set of all states $\varrho\in\state$ satisfying $\rank(\varrho)\leq r$.
Clearly, $\state^{\leq 1}\equiv\state^1$ and $\state^{\leq d}\equiv\state$. Moreover, if $d=\infty$, we denote by $\ss^{\rm fin}$ the set of states with finite rank.
In this section we investigate observables that are ($\tt$,$\pp$)-informationally complete when $\tt = \ss^{\leq t}$ and $\pp = \ss^{\leq p}$ for some $t\in\nat$, $p\in\nat_\ast$ with $1\leq t\leq p \leq d$, or $\tt = \ss^{\rm fin}$ and $\pp = \ss$, or $\tt = \pp = \ss^{\rm fin}$.

%%%%%%%%%%%%%%%%%%%%%%%%%%%%%%%%%
\subsection{Mathematical characterization}
%%%%%%%%%%%%%%%%%%%%%%%%%%%%%%%%%

By the Spectral Theorem each $T\in\trhs$ has a spectral decomposition; there exists an orthonormal basis $\{\psi_j\}_{j=1}^d$ of $\hh$ such that
\begin{equation*}
T= \sum_{j=1}^d \lambda_j \kb{\psi_j}{\psi_j} \, , 
\end{equation*}
where $\lambda_j\in\R$ and $\sum_j \lambda_j = \tr{T}$.
(In the infinite dimensional case the sum in the right hand side is infinite, in which case it converges in the trace class norm and is independent on the order of the terms.)
Each $\lambda_j$ is an eigenvalue of $T$, and every eigenvalue of $T$ appears in this decomposition as many times as is its multiplicity. Clearly, $\rank(T)$ is the number of nonzero eigenvalues of $T$ counted with their multiplicities.

In addition to $\rank(T)$, we will also need the following other characteristic numbers of $T$:
\begin{itemize}
\item $\rank_+(T)$ = number of strictly positive eigenvalues of $T$ counted with their multiplicities;
\item $\rank_-(T)$ = number of strictly negative eigenvalues of $T$ counted with their multiplicities;
\item $\rank_\uparrow(T)$ = $\max(\rank_+(T),\rank_-(T))$;
\item $\rank_\downarrow(T)$ = $\min(\rank_+(T),\rank_-(T))$.
\end{itemize}
We clearly have 
\begin{align}\label{eq:clear}
\rank_+(T) + \rank_-(T) =\rank_\uparrow(T)+\rank_\downarrow(T) = \rank(T) \leq d \, .
\end{align} 
Note also that $\rank_\pm(-T) = \rank_\mp(T)$ while $\rank_\uparrow(-T)= \rank_\uparrow(T)$ and $\rank_\downarrow(-T)= \rank_\downarrow(T)$.

We will be interested in subspaces $\mathcal{X}\subseteq \trhs$ that are annihilators of certain weak*-closed real operator systems.
By Prop. \ref{prop:exM} these subspaces consist of operators $T\in\trhs$ with $\tr{T}=0$.
The following lemma will be used later several times.

\begin{lemma}\label{prop:pos-neg}
Let $T\in\trhs$ be a nonzero operator with $\tr{T}=0$. We then have the following facts.
\begin{enumerate}[(a)]
\item The inequalities 
\begin{equation}\label{eq:pos-neg-0}
1\leq \rank_\downarrow(T)\leq \rank_\uparrow(T)\leq \rank(T)-1 \leq d-1
\end{equation}
hold.
\item There are $\rho_+,\rho_-\in\state$ and $\lambda>0$ such that
\begin{equation*}
T= \lambda (\rho_+ - \rho_- ) \, , 
\end{equation*}
and
\begin{equation*}
\rank(\rho_+)=\rank_+(T) \, , \qquad \rank(\rho_-)=\rank_-(T) \, .
\end{equation*}
\item If $\rho_1,\rho_2\in\state$ and $\lambda>0$ are such that 
\begin{equation}\label{eq:pos-neg-2}
T= \lambda (\rho_1 - \rho_2 ) \, , 
\end{equation}
then 
\begin{equation}\label{eq:rank-ineq}
\rank(\rho_1)\geq \rank_+(T) \, , \qquad \rank(\rho_2) \geq \rank_-(T) \, .
\end{equation}
\end{enumerate}
\end{lemma}

\begin{proof}
\begin{enumerate}[(a)]
\item Since $T\neq 0$, it must have a nonzero eigenvalue. Further, since $\tr{T}=0$, it must have both strictly positive and strictly negative eigenvalues, i.e.,  $\rank_\downarrow(T) \geq 1$. 
The third inequality in \eqref{eq:pos-neg-0} now follows from \eqref{eq:clear}, and the remaining inequalities are clear.

\item
Let $T=\sum_{j} \lambda_j \kb{\psi_j}{\psi_j}$ be the spectral decomposition of $T$.
{We denote by $T_+=\sum_{j\mid \lambda_j>0} \lambda_j \kb{\psi_j}{\psi_j}$ and $T_-=-\sum_{j \mid \lambda_j<0} \lambda_j \kb{\psi_j}{\psi_j}$ the positive and negative parts of $T$ respectively.}
We define $\lambda=\sum_{j\mid \lambda_j>0} \lambda_j>0$ and $\rho_{\pm}=\frac{1}{\lambda} T_{\pm}$.
Since $\tr{T}=0$, we have $\lambda=\sum_{j\mid \lambda_j<0} \lambda_j$ and $\rho_{\pm}$ thus satisfy the required conditions.

\item
We recall the following consequence of Fan's theorem {\cite[Theorem~1.7]{STIA2005}}: if two operators $A,B\in\trhs$ satisfy $A\leq B$, then $\alpha_j \leq \beta_j$ for every $j=1,2,\ldots$, where $\{\alpha_j\}$, $\{\beta_j\}$ are the eigenvalues of $A$ and $B$, respectively, ordered in the decreasing order and repeated according to their multiplicities.
From \eqref{eq:pos-neg-2} it follows that $T \leq \lambda \rho_1$ and $-T \leq \lambda \rho_2$. 
By Fan's theorem, $T$ cannot have more strictly positive eigenvalues than $\lambda \rho_1$ and thus $\rank_+(T)\leq\rank_+(\lambda \rho_1)=\rank(\rho_1)$.
Similarly, $-T$ cannot have more strictly positive eigenvalues than $\lambda \rho_2$ and thus $\rank_+(-T)\leq\rank_+(\lambda \rho_2)=\rank(\rho_2)$.
Since $\rank_+(-T)=\rank_-(T)$, we obtain \eqref{eq:rank-ineq}.
\end{enumerate}
\end{proof}

\begin{lemma}\label{prop:lemma2}
We have the following facts.
\begin{enumerate}[(a)]
\item $\R(\ss^{\leq t} - \ss^{\leq p}) = \{T\in\trhs \mid \tr{T}=0,\, \rank_\downarrow(T) \leq t \mbox{ and } \rank_\uparrow(T) \leq p\}$ for all $t,p\in\nat_\ast$ with $t\leq p\leq d$.
\item $\R(\ss^{\rm fin} - \ss^{\rm fin}) = \{T\in\trhs \mid \tr{T}=0 \mbox{ and }\rank_\uparrow(T) < \infty\}$.
\item $\R(\ss^{\rm fin} - \ss) = \{T\in\trhs \mid \tr{T}=0 \mbox{ and }\rank_\downarrow(T) < \infty\}$.
\end{enumerate}
\end{lemma}
\begin{proof}
\begin{enumerate}[(a)]
\item If $m,n\in\nat_\ast$, then
$$
\R_+ (\ss^{\leq m} - \ss^{\leq n}) = \{T\in\trhs \mid \tr{T} = 0,\, \rank_+ (T) \leq m \mbox{ and } \rank_- (T) \leq n\}
$$
as an immediate consequence of Lemma \ref{prop:pos-neg}b,c. Since
$$
\R(\ss^{\leq t} - \ss^{\leq p}) = \R_+(\ss^{\leq t} - \ss^{\leq p}) \cup \R_+(\ss^{\leq p} - \ss^{\leq t})
$$
the claim follows.
\item We have $\R(\ss^{\rm fin} - \ss^{\rm fin}) = \bigcup_{m,n\in\nat \mid m\leq n} \R(\ss^{\leq m} - \ss^{\leq n})$, hence the claim follows from (a).
\item Similarly, $\R(\ss^{\rm fin} - \ss) = \bigcup_{m\in\nat} \R(\ss^{\leq m} - \ss^{\leq d})$, which by (a) and triviality of the condition $\rank_\uparrow(T) \leq d$ implies the claim.
\end{enumerate}
\end{proof}

The following theorem characterizes $(\state^{\leq t},\state)$-informational completeness and $(\state^{\leq t},\state^{\leq p})$-informational completeness  for all values of $t$ and $p$ in both cases $d<\infty$ and $d=\infty$.

\begin{theorem}\label{prop:structure}
Let $\Mo$ be an observable and $t\in\nat$ with $1\leq t\leq d$.
\begin{enumerate}[(a)]
\item The following conditions are equivalent:
\begin{enumerate}[(i)]
\item $\Mo$ is $(\state^{\leq t},\state)$-informationally complete.  
\item Every nonzero $T\in\rr(\Mo)^\perp$ has $\rank_\downarrow(T)\geq t+1$.
\end{enumerate}
\item If $p\in\nat$ with $t\leq p \leq d$,  then the following conditions are equivalent:
\begin{enumerate}[(i)]
\item $\Mo$ is $(\state^{\leq t},\state^{\leq p})$-informationally complete.
\item Every nonzero $T\in\rr(\Mo)^\perp$ has $\rank_\downarrow(T)\geq t+1$ or $\rank_\uparrow(T)\geq p+1$.
\end{enumerate}
\end{enumerate}
\end{theorem}

\begin{proof}
These are all immediate consequences of Prop.~\ref{prop:prem} and Lemma \ref{prop:lemma2}a. For (a), note that, if $p=d$, then the condition $\rank_\uparrow(T) \leq p$ in Lemma \ref{prop:lemma2}a is trivial.
\end{proof}

If $d=\infty$, then we can also consider $(\state^{\rm fin},\state)$-informational completeness and $(\state^{\rm fin},\state^{\rm fin})$-informational completeness. 
The following theorem characterizes these two properties.

\begin{theorem}\label{prop:structure-fin}
Let $\Mo$ be an observable and $d=\infty$.
\begin{enumerate}[(a)]
\item The following conditions are equivalent:
\begin{enumerate}[(i)]
\item $\Mo$ is $(\state^{\rm fin},\state^{\rm fin})$-informationally complete.  
\item Every nonzero $T\in\rr(\Mo)^\perp$ has $\rank_\uparrow(T) = \infty$.
\end{enumerate}
\item The following conditions are equivalent:
\begin{enumerate}[(i)]
\item $\Mo$ is $(\state^{\rm fin},\state)$-informationally complete.  
\item Every nonzero $T\in\rr(\Mo)^\perp$ has $\rank_\downarrow(T) = \infty$.
\end{enumerate}
\end{enumerate}
\end{theorem}

\begin{proof}
These are all immediate consequences of Prop.~\ref{prop:prem} and Lemma \ref{prop:lemma2}b,c.
\end{proof}

With certain choices of $t$ and $p$ the conditions in Theorem~\ref{prop:structure} become simpler.
In the following we list some special cases.

Since every nonzero $T\in\trhs$ with $\tr{T}=0$ satisfies 
\begin{equation*}
2\ \rank_\uparrow(T) \geq \rank(T) \geq \rank_\uparrow(T)+1 \geq \rank_\downarrow(T)+1 \, , 
\end{equation*}
(the second inequality following from Lemma \ref{prop:pos-neg}a) we recover the following two consequences of Theorem~\ref{prop:structure}, first presented in \cite{HeMaWo13}.

\begin{corollary}\label{prop:S1S1}
Let $\Mo$ be an observable. The following conditions are equivalent:
\begin{enumerate}[(i)]
\item $\Mo$ is $(\state^1,\state^1)$-informationally complete.
\item Every nonzero $T\in\rr(\Mo)^\perp$ satisfies $\rank(T)\geq 3$.
\end{enumerate}
\end{corollary}

\begin{corollary}\label{prop:SpSp}
Let $\Mo$ be an observable and $t\in\nat$ such that $1\leq t\leq d$.
The following conditions are equivalent:
\begin{enumerate}[(i)]
\item $\Mo$ is $(\state^{\leq t},\state^{\leq t})$-informationally complete.  
\item Every nonzero $T\in\rr(\Mo)^\perp$ has $\rank_\uparrow(T)\geq t+1$.
\end{enumerate}
A necessary condition for these equivalent conditions is that every nonzero $T\in\rr(\Mo)^\perp$ satisfies $\rank(T)\geq t+2$, and a sufficient condition is that every nonzero $T\in\rr(\Mo)^\perp$ satisfies $\rank(T)\geq 2t+1$.
\end{corollary}

%%%%%%%%%%%%%%%%%%%%%%%%%%%%%%%%%
\subsection{Equivalent and inequivalent properties}\label{sec:equivalent}
%%%%%%%%%%%%%%%%%%%%%%%%%%%%%%%%%

Obviously, the set $\state^{\leq r_1}$ contains $\state^{\leq r_2}$ whenever $r_1\geq r_2$.
It follows that a $(\state^{\leq t_1},\state^{\leq p_1})$-informationally complete observable is also $(\state^{\leq t_2},\state^{\leq p_2})$-informationally complete for all $t_2\leq t_1$ and $p_2\leq p_1$.
Physically speaking, smaller $t$ means easier task while smaller $p$ means stronger premise, hence the relation between the above properties is easy to understand. Moreover, if $d=\infty$, then $(\state^{\rm fin},\state)$-informational completeness implies $(\state^{\rm fin},\state^{\rm fin})$-informational completeness, and for the different kinds of informational completeness we have the following implications:
\begin{gather*}
(\state^{\leq t},\state)   \\
\Nearrow \qquad \Searrow  \\
\qquad (\state,\state) \Rightarrow (\state^{\rm fin},\state)   \qquad \qquad  (\state^{\leq t},\state^{\leq p}) \Rightarrow (\state^{\leq t},\state^{\leq t}) \\
\Searrow \qquad \Nearrow  \\
 (\state^{\rm fin},\state^{\rm fin}) 
\end{gather*}
for $t,p\in\nat$ with $1\leq t\leq p \leq d$.

For a fixed dimension $d<\infty$, there are $\half d (d+1)$ pairs $(t,p)$ consisting of integers with $1\leq t\leq p \leq d$.
We would thus expect to have $\half d (d+1)$ different properties of $(\state^{\leq t},\state^{\leq p})$-informational completeness.
But as we will see, for some values of $t_1,p_1$ and $t_2,p_2$ the corresponding properties are equivalent.
We will derive a complete classification of the inequivalent properties and it turns out that there are only
\[
\bigfloor{\frac{d}{2}} \left(d-\frac{1}{2}-\frac{1}{2} \bigfloor{\frac{d}{2}}\right)
\]
inequivalent forms of $(\state^{\leq t},\state^{\leq p})$-informational completeness.

\begin{figure}
\begin{center}
\includegraphics[scale=0.35]{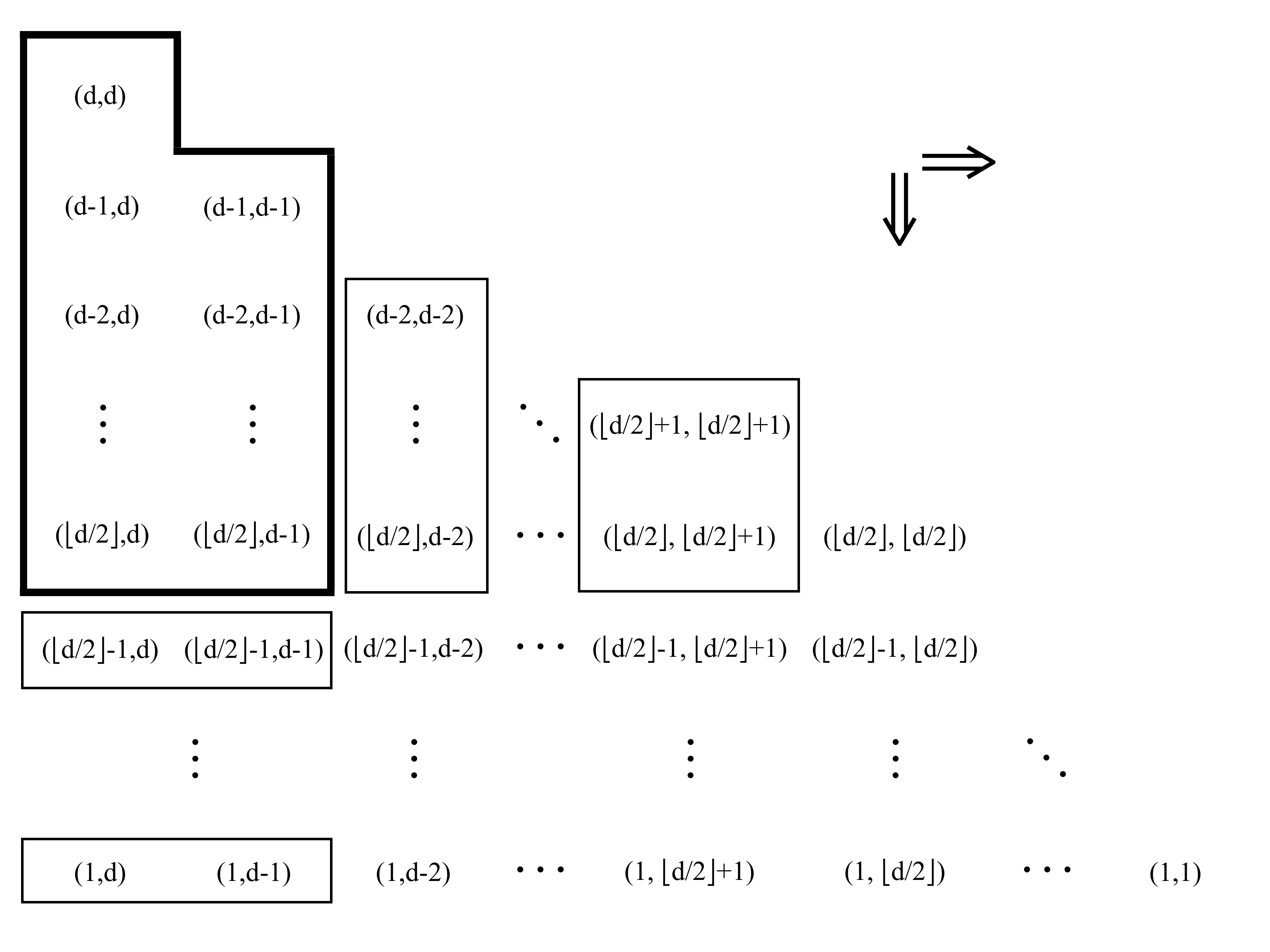}
\end{center}
\caption{In this picture $d<\infty$. 
Each $(t,p)$ represents the property of $(\state^{\leq t},\state^{\leq p})$-informational completeness.
One property implies another one if the latter can be reached from the first by moving down and right. Equivalent properties are in the same box. The box with thick boundary is the set of all properties that are equivalent to informational completeness.}
\label{fig:equivalence}
\end{figure}

\begin{proposition}\label{prop:equivalence-p}
(Equivalence of different premises.)
Let $2 \leq d < \infty$. Let $t\in\nat$ be such that $1\leq t \leq d-1$.
The following properties are equivalent:
\begin{enumerate}[(i)]
\item $(\state^{\leq t},\state^{\leq d-1})$-informational completeness.
\item $(\state^{\leq t},\state)$-informational completeness.
\end{enumerate}
\end{proposition}

\begin{proof}
It is clear that (ii)$\Rightarrow$(i).
To show that (i)$\Rightarrow$(ii), assume that $\Mo$ is a $(\state^{\leq t},\state^{\leq d-1})$-informationally complete observable.
By Theorem~\ref{prop:structure}b this means that every nonzero $T\in\rr(\Mo)^\perp$ has $\rank_\downarrow(T)\geq t+1$ or $\rank_\uparrow(T)\geq d$.
But the second condition cannot hold since $\rank_\uparrow(T) \leq d-1$ by Lemma \ref{prop:pos-neg}a.
Therefore, every nonzero $T\in\rr(\Mo)^\perp$ has $\rank_\downarrow(T)\geq t+1$.
Then, Theorem~\ref{prop:structure}a shows that $\Mo$ is $(\state^{\leq t},\state)$-informationally complete.
\end{proof}

\begin{proposition}\label{prop:equivalence-t}
(Equivalence of different tasks.)
Let $2 \leq d < \infty$. 
Let $t,p\in\nat$ be such that $\floor{\frac{d}{2}} \leq t\leq p\leq d$.
The following properties are equivalent:
\begin{enumerate}[(i)]
\item $(\state^{\leq \floor{\frac{d}{2}}},\state^{\leq p})$-informational completeness.
\item $(\state^{\leq t},\state^{\leq p})$-informational completeness.
\item $(\state^{\leq p},\state^{\leq p})$-informational completeness.
\end{enumerate}
\end{proposition}

\begin{proof}
It is clear that (iii)$\Rightarrow$(ii)$\Rightarrow$(i).
To show that (i)$\Rightarrow$(iii), assume that $\Mo$ is a $(\state^{\leq \floor{\frac{d}{2}}},\state^{\leq p})$-informationally complete observable.
By Theorem~\ref{prop:structure}b this means that every nonzero $T\in\rr(\Mo)^\perp$ has $\rank_\downarrow(T)\geq \floor{\frac{d}{2}}+1$ or $\rank_\uparrow(T)\geq p+1$.
But the first condition cannot hold since $\rank_\downarrow(T)\leq \floor{\frac{d}{2}}$ for every {nonzero} selfadjoint operator $T$.
Therefore, every nonzero $T\in\rr(\Mo)^\perp$ has $\rank_\uparrow(T)\geq p+1$.
Then, by Cor.~\ref{prop:SpSp} $\Mo$ is $(\state^{\leq p},\state^{\leq p})$-informationally complete.
\end{proof}

\begin{example}(\emph{Dimension $2$}.)\label{ex:2}
Let $d=2$.
By Prop.~\ref{prop:equivalence-p} and Prop.~\ref{prop:equivalence-t} the property of $(\state^{\leq t},\state^{\leq p})$-informational completeness is equivalent to informational completeness for the all three possible pairs $(t,p)$: $(2,2)$, $(1,2)$, $(1,1)$.
\end{example}

\begin{proposition}\label{prop:inequivalence-p}
(Inequivalence of different premises.)
Let $3 \leq d \leq \infty$. 
Let $t,p_1\in\nat$, $p_2\in\nat_\ast$ be such that $1\leq t \leq p_1< p_2 \leq d-1$. 
The following properties are not equivalent:
\begin{enumerate}[(i)]
\item $(\state^{\leq t},\state^{\leq p_1})$-informational completeness.
\item $(\state^{\leq t},\state^{\leq p_2})$-informational completeness.
\end{enumerate}
\end{proposition}

\begin{proof}
Fix an orthonormal basis $\{\psi_j\}_{j=1}^d$ and define an operator $T$ by
$$
T = \frac{1}{p_1+1}\sum_{j=1}^{p_1+1}\kb{\psi_j}{\psi_j} - \kb{\psi_{p_1+2}}{\psi_{p_1+2}} \, .
$$
It follows from $p_1 < p_2 \leq d-1$ that $p_1\leq d-2$, hence, as we are also assuming $p_1\in\nat$, this definition makes sense.
Since $T^\ast=T$ and $\tr{T}=0$, we conclude from Prop. \ref{prop:exM} that there exists an observable $\Mo$ such that $\rr(\Mo)^\perp = \R T$. 
As $\rank_\uparrow(\lam T)=p_1+1$ and $\rank_\downarrow(\lam T)=1$ for every $\lam\neq 0$, it follows from Theorem~\ref{prop:structure}b that $\Mo$ is $(\state^{\leq t},\state^{\leq p_1})$-informationally complete, but not $(\state^{\leq t},\state^{\leq p_2})$-informationally complete by Theorem~\ref{prop:structure}b (if $p_2\in\nat$) or Theorem~\ref{prop:structure}a (if $p_2 = d = \infty$).
\end{proof}

\begin{example}(\emph{Dimension $3$}.)\label{ex:3}
Let $d=3$.
Using Prop.~\ref{prop:equivalence-p} and Prop.~\ref{prop:equivalence-t} we see that the property of $(\state^{\leq t},\state^{\leq p})$-informational completeness is equivalent to informational completeness for five choices of $(t,p)$: $(3,3)$, $(2,3)$, $(1,3)$, $(2,2)$ and $(1,2)$.
The remaining property, namely  $(\state^{1},\state^{1})$-informational completeness, is not equivalent to $(\state^{1},\state^{\leq 2})$-informational completeness (and hence not to any other) by Prop.~\ref{prop:inequivalence-p}.
\end{example}

\begin{proposition}\label{prop:inequivalence-t}
(Inequivalence of different tasks.)
Let $4 \leq d \leq \infty$.
Let $t_1,t_2\in\nat$ and $p\in\nat_\ast$ such that $1\leq t_1<t_2 \leq p\leq d$ and $t_2 \leq \floor{\frac{d}{2}}$.
The following properties are not equivalent:
\begin{enumerate}[(i)]
\item $(\state^{\leq t_1},\state^{\leq p})$-informational completeness.
\item $(\state^{\leq t_2},\state^{\leq p})$-informational completeness.
\end{enumerate}
\end{proposition}

\begin{proof}
Fix an orthonormal basis $\{\psi_j\}_{j=1}^d$ and define an operator $T$ by
$$
T=\frac{1}{t_1+1} \sum_{j=1}^{t_1+1}\kb{\psi_j}{\psi_j} - \frac{1}{t_1+1} \sum_{j=t_1+2}^{2t_1+2}\kb{\psi_j}{\psi_j} \, .
$$
It follows from $t_1 < t_2 \leq \floor{\frac{d}{2}}$ that $t_1\leq \floor{\frac{d}{2}}-1$, hence $2t_1+2\leq d$ and, as we are also assuming $t_1\in\nat$, this definition makes sense.
Since $T^\ast=T$ and $\tr{T}=0$, we conclude from Prop.~\ref{prop:exM} that there exists an observable $\Mo$ such that $\rr(\Mo)^\perp = \R T$. 
We have $\rank_\downarrow(\lam T)=\rank_\uparrow(\lam T)=t_1+1$ for every $\lam\neq 0$.
By Theorem~\ref{prop:structure}b (if $p\in\nat$) or Theorem~\ref{prop:structure}a (if $p = d = \infty$), $\Mo$ is $(\state^{\leq t_1},\state^{\leq p})$-informationally complete but not $(\state^{\leq t_2},\state^{\leq p})$-informationally complete.
\end{proof}

\begin{example}(\emph{Dimension $d=4$}.)\label{ex:4}
Let $d=4$.
Using Prop. \ref{prop:equivalence-p} and Prop.~\ref{prop:equivalence-t} we see that the property of $(\state^{\leq t},\state^{\leq p})$-informational completeness is equivalent to informational completeness for five choices of $(t,p)$: $(4,4)$, $(3,4)$, $(2,4)$, $(3,3)$ and $(2,3)$.
We also see that the properties corresponding to $(1,4)$ and $(1,3)$ are equivalent but inequivalent to informational completeness.
By Prop.~\ref{prop:inequivalence-p} and Prop.~\ref{prop:inequivalence-t}, the remaining properties corresponding to $(2,2)$, $(1,2)$ and $(1,1)$ are not equivalent to any other choices of $(t,p)$.
\end{example}

A moment's thought shows that Props.~\ref{prop:equivalence-p} - \ref{prop:inequivalence-t} give a complete classification of the $(\state^{\leq t},\state^{\leq p})$-informational completeness properties into equivalent and inequivalent collections in all finite dimensions.
The classification for $4 \leq d < \infty$ is summarized in Fig.~\ref{fig:equivalence}. 
One of the most interesting consequences of this classification is the following.

\begin{corollary}
(Equivalence to informational completeness.)
Let $2 \leq d < \infty$.
For integers $1\leq t \leq p \leq d$, $(\state^{\leq t},\state^{\leq p})$-informational completeness is equivalent to informational completeness if and only if $p\geq d-1$ and $t\geq \floor{\frac{d}{2}}$.
\end{corollary}

This result is implying, in particular, that if there is no premise and the task is to determine all states with rank less or equal to $\floor{\frac{d}{2}}$, then we actually need an informationally complete observable.

When $d=\infty$, the infinite dimensional state space has also proper subsets that do not exist in the finite dimensional case, hence leading to new kinds of tasks and premises.
For example, we can have a premise that the system has finite rank but we do not know any upper bound for its rank, i.e., $\pp = \ss^{\rm fin}$.
Similarly, we could be interested in the task of determining all states with finite rank, i.e., $\tt = \ss^{\rm fin}$, when $\pp = \ss$. We already characterized $(\ss^{\rm fin},\ss^{\rm fin})$- and $(\ss^{\rm fin},\ss)$-informational completeness in Theorem~\ref{prop:structure-fin}.
The next proposition shows that these two properties are not the same.

\begin{proposition}\label{prop:Sf}
Let $d = \infty$.
The following properties are all inequivalent:
\begin{enumerate}[(i)]
\item $(\state^{\rm fin},\state^{\rm fin})$-informational completeness.
\item $(\state^{\rm fin},\state)$-informational completeness.
\item informational completeness.
\end{enumerate}
\end{proposition}

\begin{proof}
Clearly, (iii) $\Rightarrow$ (ii) $\Rightarrow$ (i). Hence we need to show that (i) $\nRightarrow$ (ii) and (ii) $\nRightarrow$ (iii). There clearly exist  $T_1,T_2\in \trhs$ such that $\tr{T_i}=0$, $\rank_+(T_1)<\infty$ and $\rank_-(T_1)=\rank_\pm(T_2)=\infty$. By Prop.~\ref{prop:exM}, there exists two observables $\Mo_1$ and $\Mo_2$ such that $\rr(\Mo_i)^\perp = \R T_i$. As in the proofs of Props.~\ref{prop:inequivalence-p} and \ref{prop:inequivalence-t}, we have $\rank_\downarrow(T_1') < \infty$ and $\rank_\uparrow(T_1') = \infty$ for all nonzero $T_1'\in\rr(\Mo_1)^\perp$, and $\rank_\downarrow(T_2') = \rank_\uparrow(T_2') = \infty$ for all nonzero $T_2'\in\rr(\Mo_2)^\perp$. Thus, by Theorem~\ref{prop:structure-fin} the observable $\Mo_1$ is $(\ss^{\rm fin},\ss^{\rm fin})$-informationally complete but not $(\ss^{\rm fin},\ss)$-informationally complete. Similarly, by Theorem~\ref{prop:structure-fin} and Cor.~\ref{prop:well-known} the observable $\Mo_2$ is $(\ss^{\rm fin},\ss)$-informationally complete but not informationally complete.
\end{proof}

The content of Prop. \ref{prop:Sf} is, essentially, that knowing that the unknown state has finite rank is useful information for state determination.

%%%%%%%%%%%%%%%%%%%%%%%%%%%%
\section{Minimal number of outcomes}\label{sec:min}
%%%%%%%%%%%%%%%%%%%%%%%%%%%%

In this section we assume that $d<\infty$ and $\# \Omega < \infty$.

%%%%%%%%%%%%%%%%%%%%%%%%%%%%
\subsection{General formulation of the problem}
%%%%%%%%%%%%%%%%%%%%%%%%%%%%

By a \emph{minimal $(\task,\prem)$-informationally complete observable} we mean a $(\task,\prem)$-informationally complete observable with minimal number of outcomes. More precisely, an observable $\Mo$ with an outcome space $\Omega$ is minimal $(\task,\prem)$-informationally complete if any other $(\task,\prem)$-informationally complete observable $\Mo'$ with an outcome space $\Omega'$ satisfies $\#\Omega \leq \#\Omega'$. 

Since $\hi$ is finite dimensional, the real vector spaces $\lhs$ and $\trhs$ are the same and
\begin{equation}\label{eq:dimsum}
\dim \rr + \dim \rr^\perp = \dim\lhs = d^2 \, .
\end{equation}
We then see that a $(\task,\prem)$-informationally complete observable $\Mo$ with $n$ outcomes exists if and only if there is a $(d^2 - n)$-dimensional subspace $\xx \subseteq\trhs$ satisfying
\begin{enumerate}[(1)]
\item $\tr{T} = 0$ for all $T\in\xx$;
\item $\xx \cap \ \R(\task-\prem)  = \{0\}$.
\end{enumerate}
Indeed, in this case by Prop.~\ref{prop:exM} we can find an observable $\Mo$ with $\rr(\Mo)^\perp = \xx$ and having $d^2 - \dim\rr(\Mo)^\perp = n$ outcomes. Such an observable is $(\task,\prem)$-informationally complete by Prop.~\ref{prop:prem}. We thus conclude that seeking a minimal $(\task,\prem)$-informationally complete observable is equivalent to looking for a real subspace $\xx\subseteq\trhs$ satisfying (1) and (2), and having maximal dimension among all subspaces of $\trhs$ with the properties (1) and (2).
Once such a maximal subspace $\xx$ is found, then the minimal number of outcomes for a $(\task,\prem)$-informationally complete observable is $d^2 - \dim\xx$.

%%%%%%%%%%%%%%%%%%%%%%%%%%%%
\subsection{Review of some known bounds}
%%%%%%%%%%%%%%%%%%%%%%%%%%%%

If $d=2$, then all $(\state^{\leq t},\state^{\leq p})$-informational completeness properties are equivalent (see Example \ref{ex:2}), hence the minimal number of outcomes is $d^2=4$ in all of them.

If $d=3$, then only $(\state^{1},\state^{1})$-informationally completeness is inequivalent to informational completeness (see Example \ref{ex:3}).
In the latter case the minimal number is $d^2=9$, while in the first case a simple argument shows that the minimal number of outcomes is $8$; see Prop. 5 in \cite{HeMaWo13}.

Let us then assume $4 \leq  d<\infty$ and recall some bounds for the minimal number of $(\state^{\leq t},\state^{\leq t})$- and $(\state^{\leq t},\state)$-informationally complete observables.
In these cases we need to find subspaces $\xx$ such that every nonzero $T\in\xx$ satisfies $\tr{T} = 0$ and $\rank_\uparrow(T)\geq t+1$ (Cor.~\ref{prop:SpSp}) or $\rank_\downarrow(T)\geq t+1$ (Theorem~\ref{prop:structure}a), respectively.
To find a good upper bound for the minimal number of outcomes, we need to find as large $\xx$ as possible.
A useful method for constructing these kind of subspaces was presented in \cite{CuMoWi08}. 
Using this method the following upper bounds (a) and (b) were proved in \cite{HeMaWo13} and \cite{ChDaJiJoKrShZe}, respectively.

\begin{proposition}\label{prop:upper}
Let $1\leq t < d/2$.
There exists 
\begin{enumerate}[(a)]
\item $(\state^{\leq t},\state^{\leq t})$-informationally complete observable with $4t(d-t)$ outcomes. 
\item $(\state^{\leq t},\state)$-informationally complete observable with $4t(d-t)+d-2t$ outcomes. 
\end{enumerate}
\end{proposition}

In the case of $(\state^{\leq t},\state^{\leq t})$-informationally complete observables, it is possible to obtain lower bounds from the known non-embedding results for Grassmannian manifolds \cite{HeMaWo13}.
In some cases the obtained lower bounds agree or are very close with the upper bounds written in Prop. \ref{prop:upper}a.
In particular, it was proved in \cite{HeMaWo13} that  in the case of $(\state^{1},\state^{1})$-informational completeness, the minimal number of outcomes is not a linear function of $d$ but differs from the upper bound $4d-4$ at most $2 \log_2(d)$.
Also a slightly better upper bound was derived, and these results give the exact answer for many $d$. 
For instance, for the dimensions between $2$ and $100$, the results of  \cite{HeMaWo13} give the exact minimal number in $45$ cases.

In the case of $(\state^{1},\state)$-informational completeness, the upper bound for the minimal number of outcomes is $5d-6$ \cite{ChDaJiJoKrShZe}.
Obviously, the known lower bound for minimal $(\state^{1},\state^{1})$-informational completeness is also a lower bound for minimal $(\state^{1},\state)$-informational completeness. 
We are not aware of any better lower bound.
In the following subsection we prove that the minimal number of outcomes for $d=4$ is $11$.
This means that $5d-6$ is generally just an upper bound for the minimal number of outcomes, not the exact answer.
Our result for $d=4$ also implies that, as in the case of $(\state^{1},\state^{1})$-informational completeness, the minimal number is not a linear function of $d$.

%%%%%%%%%%%%%%%%%%%%%%
\subsection{Dimension 4}\label{sec:4}
%%%%%%%%%%%%%%%%%%%%%%

In this subsection we concentrate on minimal observables in dimension $4$.
A minimal informationally complete observable has $d^2=16$ outcomes.
Further, it was shown in \cite{HeMaWo13} that a minimal $(\state^1,\state^1)$-informationally complete observable has $10$ outcomes.
In Prop. \ref{prop:d=4/2} below we give the minimal numbers for the remaining three inequivalent properties (see Example \ref{ex:4}). 
These results are summarized in Fig. \ref{fig:equivalence-4}.
Before deriving the minimal numbers we characterize these properties in convenient forms.

\begin{figure}
\begin{center}
\includegraphics[scale=0.3]{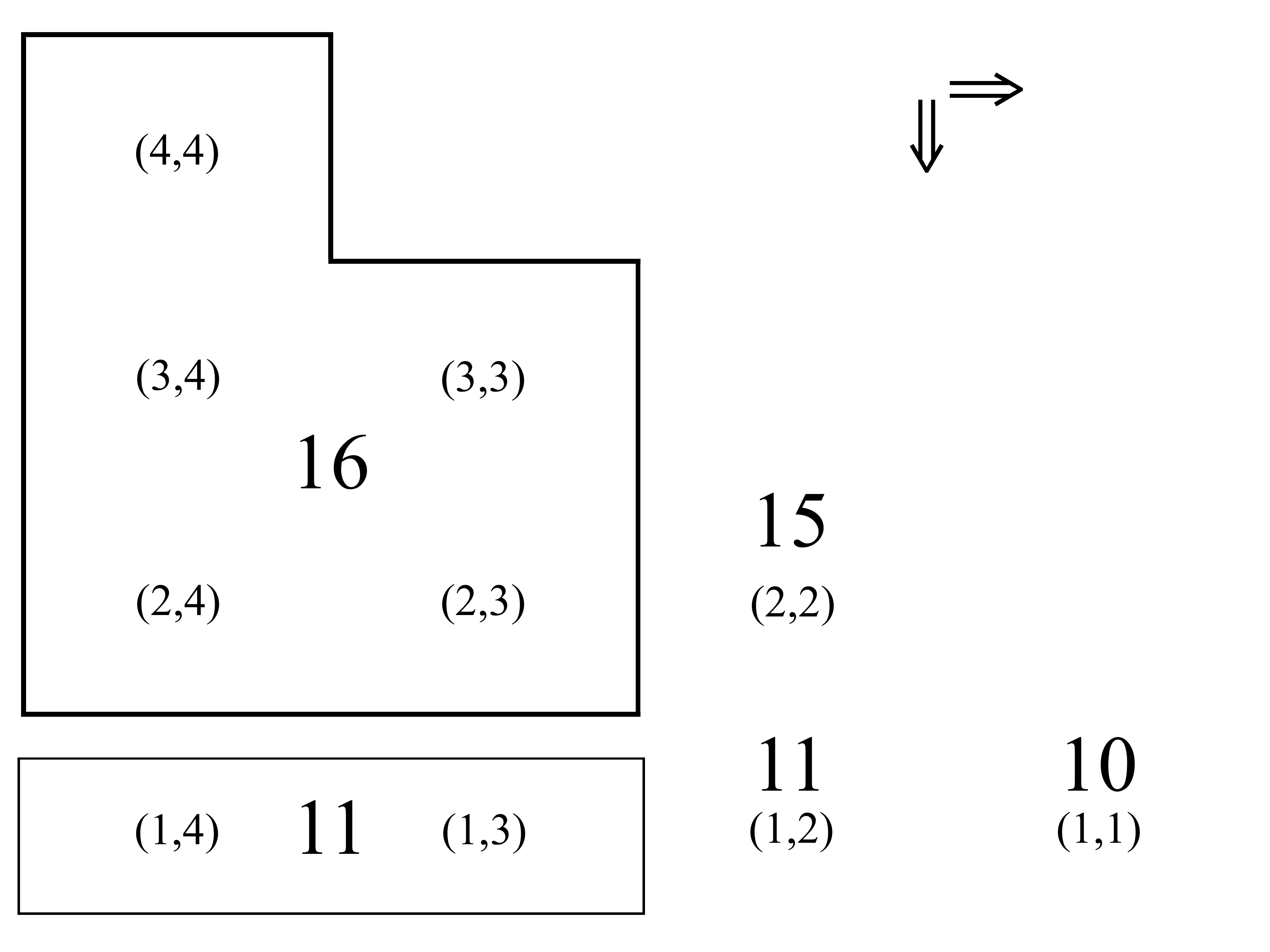}
\end{center}
\caption{In this picture $d=4$.
Each $(t,p)$ represents the property of $(\state^{\leq t},\state^{\leq p})$-informational completeness, and equivalent properties are in the same box.
As explained in Example \ref{ex:4}, there are five inequivalent properties. 
The big numbers give the minimal number of outcomes that an $(\state^{\leq t},\state^{\leq p})$-informationally complete observable must have.}
\label{fig:equivalence-4}
\end{figure}

\begin{proposition}\label{prop:d=4/1}
Let $d=4$. An observable $\Mo$ is 
\begin{enumerate}[(a)]
\item $(\state^1,\state)$-informationally complete if and only if every nonzero $T\in \rr(\Mo)^\perp$ satisfies $\det{T}>0$. 
\item $(\state^{\leq 2},\state^{\leq 2})$-informationally complete if and only if every nonzero $T\in \rr(\Mo)^\perp$ satisfies $\det{T}<0$. 
\item $(\state^1,\state^{\leq 2})$-informationally complete if and only if every nonzero $T\in \rr(\Mo)^\perp$ satisfies $\det{T}\neq 0$. 
\end{enumerate}
If an observable is $(\state^1,\state^{\leq 2})$-informationally complete, then it is either $(\state^1,\state)$-informationally complete or $(\state^{\leq 2},\state^{\leq 2})$-informationally complete.
\end{proposition}

\begin{proof}
By Lemma \ref{prop:pos-neg}a every nonzero $T\in\rr(\Mo)^\perp$ has $1 \leq \rank_\pm(T)\leq 3$.
Since $\det{T}$ is the product of eigenvalues, we conclude that every nonzero $T\in \rr(\Mo)^\perp$ satisfies 
\begin{enumerate}[(a)]
\item $\det{T}>0$ if and only if every nonzero $T\in \rr(\Mo)^\perp$ satisfies $\rank_\downarrow(T)=2$.
\item $\det{T}<0$ if and only if every nonzero $T\in \rr(\Mo)^\perp$ satisfies $\rank_\uparrow(T)=3$.
\item $\det{T}\neq 0$ if and only if every nonzero $T\in \rr(\Mo)^\perp$ satisfies $\rank_\downarrow(T)=2$ or $\rank_\uparrow(T)=3$.
\end{enumerate}
The claims (a), (b) and (c) in Prop. \ref{prop:d=4/1} now follow from Theorem~\ref{prop:structure}a, Cor.~\ref{prop:SpSp} and Theorem~\ref{prop:structure}b, respectively.

To prove the last claim, suppose $\xx\subseteq \trhs$ is a subspace such that every nonzero $X\in\xx$ satisfies $\det{X}\neq 0$. 
We need to prove that the sign of $\det{X}$ is constant for all nonzero $X\in \xx$.
If $\dim\xx=1$, then this is clearly true.
So assume that $\dim\xx\geq 2$.
We make a counter assumption: $X,Y\in\xx$ are two linearly independent matrices with $\det{X}<0$ and $\det{Y}>0$.
Then $tX+(1-t)Y\in\xx\setminus\{0\}$ for every $t\in\R$, and $\det{t_0X+(1-t_0)Y}=0$ for some $0<t_0<1$ by the intermediate value theorem.
\end{proof}

\begin{proposition}\label{prop:d=4/2}
Let $d=4$.
\begin{enumerate}[(a)]
\item A minimal $(\state^1,\state)$-informationally complete observable has $11$ outcomes.
\item A minimal $(\state^{\leq 2},\state^{\leq 2})$-informationally complete observable has $15$ outcomes.
\item A minimal $(\state^1,\state^{\leq 2})$-informationally complete observable has $11$ outcomes.
\end{enumerate}
\end{proposition}

\begin{proof}
(a) For all $n\in\nat$, denote by $M_n(\C)$ the complex linear space of $n\times n$ complex matrices, and by $M_n(\C)_s$ the real space of selfadjoint elements in $M_n(\C)$. By Prop.~\ref{prop:d=4/1} we need to look for subspaces $\xx\subseteq M_4 (\C)_s$ such that
\begin{enumerate}[(1)]
\item  $\tr{X} =0$ for every $X\in\xx$;
\item $\det{X} > 0$ for every nonzero $X\in\xx$.
\end{enumerate}
Indeed, if $\xx$ has maximal dimension among all subspaces of $M_4 (\C)_s$ satisfying these two conditions, then any observable $\Mo$ with $\rr(\Mo)^\perp = \xx$ and $d^2 - \dim\xx$ outcomes (which exists by Prop.~\ref{prop:exM} and \eqref{eq:dimsum}) is minimal $(\state^1,\state)$-informationally complete.
It was shown in \cite{ALP66} that the maximal dimension of a real subspace $\xx\subseteq M_4 (\C)_s$ satisfying (2) is $5$.
We show that if the additional requirement (1) is added, this maximal dimension remains the same, and thus the minimal number of outcomes is $4^2 - 5 = 11$.
To do this, we introduce four $2\times 2$ matrices
\begin{align*}
& \sigma^1=\lft
\begin{array}{cc}
0 & 1 \\
1 & 0
\end{array}
\rgt
\, , \quad 
\sigma^2=\lft
\begin{array}{cc}
0 & -i \\
i & 0
\end{array}
\rgt \, ,  \\
& \sigma^3=\lft
\begin{array}{cc}
1 & 0 \\
0 & -1
\end{array}
\rgt \, , \quad
\sigma^4=\lft
\begin{array}{cc}
i & 0 \\
0 & i
\end{array}
\rgt \, ,
\end{align*}
and define the following linear map $N:\R^4 \to M_2 (\C)$
$$
N(\vec{a}) = \sum_{i=1}^4 a_i \sigma^i \qquad \forall \vec{a} = (a_1,a_2,a_3,a_4)\in\R^4 \, .
$$
Note that
$$
N(\vec{a})^* N(\vec{a}) = N(\vec{a})N(\vec{a})^* = \no{\vec{a}}^2 \id \qquad \forall \vec{a}\in\R^4 \, .
$$
Next, we define five $4\times 4$ selfadjoint matrices
\begin{align*}
& A_0=\lft
\begin{array}{cc}
\id & 0\\
 0 & -\id
\end{array}
\rgt 
\, , \qquad
 A_i=\lft
\begin{array}{cc}
0& \sigma^i \\
\sigma^{i\,\ast} & 0
\end{array}
\rgt
\quad \mbox{for } i\in\{1,2,3,4\} 
\end{align*}
and the following linear map $N':\R^5 \to M_4 (\C)_s$
$$
N'(a_0,\vec{a}) = \sum_{i=0}^4 a_i A_i =
\lft
\begin{array}{cc}
a_0\id& N(\vec{a})\\
N(\vec{a})^* & -a_0\id
\end{array}
\rgt \, .
$$
Clearly, $\tr{N'(a_0,\vec{a})} = 0$ for all $(a_0,\vec{a})$. 
Moreover, it is easy to verify that the matrices $A_0,\ldots,A_4$ are linearly independent.
It follows that the map $N'$ is injective, hence the image $N'(\R^5)$ of $N'$ is a $5$-dimensional subspace of $M_4 (\C)_s$. 
Finally, using the formula for the determinant of square block matrices {\cite[Theorem~3]{SDBM2000}} we obtain
$$
\det{N'(a_0,\vec{a})} = \det{-a_0^2\id - N(\vec{a})N(\vec{a})^*} = (a_0^2 + \no{\vec{a}}^2)^2 \geq 0
$$
and  $\det{N'(a_0,\vec{a})}=0$ if and only if $a_0 = 0$ and $\vec{a}=0$. Thus, $\xx=N'(\R^5)$ has the required properties.

(b) Suppose  $\xx\neq\{0\}$ is a real subspace of $\trhs$ such that $\det{T}<0$ for all nonzero $T\in\xx$. We claim that $\dim\xx = 1$. 
To prove this, let us first make a counter assumption that $\dim\xx \geq 2$.
We fix two linearly independent $X,Y\in\xx$, and then the map
\begin{align*}
& r:\spannoR{X,Y}\setminus\{0\} \to \Z \\
& r(T) = \rank_+(T) - \rank_-(T) = \tr{T^{-1}|T|}
\end{align*}
is continuous by the continuity of each map $T\mapsto T^{-1}$ and $T\mapsto |T|$. 
Since the set $\spannoR{X,Y}\setminus\{0\}$ is connected, this would imply that $r$ is constant, hence $r(T) = r(-T)$. 
It follows that $\rank_+(T) = \rank_-(T) = 2$ for all $T\in\spannoR{X,Y}\setminus\{0\}$, but this is impossible as $\det{T}<0$. 
Thus, the counter assumption is false.

We still need to prove that there exists a $1$-dimensional subspace of $\trhs$ such that $\det{T}<0$ and $\tr{T} = 0$ for all nonzero $T\in\xx$.
Fix an orthonormal basis $\{\varphi_j\}_{j=1}^4$ of $\hh$, and set
$$
T = \frac{1}{3}\sum_{j=1}^{3}\kb{\varphi_j}{\varphi_j} - \kb{\varphi_4}{\varphi_4} \, .
$$
Then, $\xx=\R T$ is a $1$-dimensional subspace of $\trhs$ such that $\det{T}<0$ and $\tr{T} = 0$ for all nonzero $T\in\xx$. Thus, there exists an observable $\Mo$ with $4^2-1 = 15$ outcomes such that $\xx=\rr(\Mo)^\perp$, and such observable is minimal $(\state^{\leq 2},\state^{\leq 2})$-informationally complete by Props.~\ref{prop:exM} and \ref{prop:d=4/1}.

(c) This follows from Prop. \ref{prop:d=4/1} combined with items (a) and (b) above.
\end{proof}

%%%%%%%%%%%
\section{Covariant phase space observables}
%%%%%%%%%%%

We now turn our attention to covariant phase space observables. After introducing these observables in the general case of a phase space defined by an Abelian group, we treat the finite and infinite dimensional cases separately. 
As an application we study the effect of noise on the observable's ability to perform the required state determination tasks.

%%%%%%%%%%%
\subsection{General formalism}
%%%%%%%%%%%

Let $\gg$ be a locally compact and second countable Abelian group with the dual group $\ggh$. The composition laws of $\gg$ and $\ggh$ will be denoted by addition, and the canonical pairing of $x\in \gg$ and $\xi\in\ggh$ will be denoted by $\pair{\xi}{x}$. 
We fix Haar measures $\de x$ and $\de \xi$ on $\gg$ and $\ggh$, respectively. If $\mu$ is any bounded measure on $\gg\times\ggh$, the {\em symplectic Fourier transform} of $\mu$ is the bounded continuous function $\widehat{\mu}$ on $\gg\times\ggh$ given by
$$
\widehat{\mu} (x,\xi) = \int \overline{\pair{\zeta}{x}} \pair{\xi}{y} \de\mu (y,\zeta) \, .
$$
This definition clearly extends to any integrable function: if $f\in L^1(\gg\times\ggh)$, we define
$$
\widehat{f} (x,\xi) = \int \overline{\pair{\zeta}{x}} \pair{\xi}{y} f(y,\zeta) \de y \de\zeta \, .
$$
For the rest of this section, we will assume that the Haar measures $\de x$ and $\de \xi$ are normalized so that $(\widehat{f})^\wedge (x,\xi) = f(x,\xi)$ whenever also $\widehat{f}\in L^1(\gg\times\ggh)$.

Let $\hh=L^2 (\gg)$. We define the following two unitary representations $U$ and $V$ of $\gg$ and $\ggh$ on $\hh$
$$
[U(x) \psi] (y) = \psi(y-x) \, , \qquad [V(\xi) \psi] (y) = \pair{\xi}{y} \psi(y) \, .
$$
Note that
$$
V(\xi)U(x) = \pair{\xi}{x} U(x)V(\xi) \, ,
$$
so that the following {\em Weyl map}
$$
W:\gg \times\ggh \to \lh \qquad W(x,\xi) = U(x)V(\xi)
$$
is a projective square integrable representation of the direct product group $\gg \times\ggh$ on $\hh$ (for square integrability of $W$, see e.g.~\cite{Werner84} in the case $\gg=\ggh=\R^n$, and \cite[Theorem~6.2.1]{Gro98} and \cite{KiSc20??} for the general case).
The Weyl map has the useful properties
\begin{align}\label{eq:J**}
W(x,\xi)W(y,\zeta) & = \pair{\xi}{y} W(x+y,\xi+\zeta) \nonumber \\
& = \pair{\xi}{y} \overline{\pair{\zeta}{x}} W(y,\zeta)W(x,\xi)
\end{align}
and
\begin{align}\label{eq:-x-xi}
W(-x,-\xi) = \overline{\pair{\xi}{x}} W(x,\xi)^* \, .
\end{align}

For any set $X\subset\gg \times\ggh$, we denote
$$
W(X) = \overline{\spannoC{W(x,\xi)\mid (x,\xi)\in X}}^{\mathrm{w^*}} \, .
$$
If $X$ is a symmetric set, i.e., 
\begin{equation*}
X \equiv -X = \{(x,\xi)\in\gg \times\ggh \mid (-x,-\xi) \in X\} \, , 
\end{equation*}
then from \eqref{eq:-x-xi} it follows that $W(X)^\ast = W(X)$, and we can thus consider the selfadjoint part $W(X)_s$ of $W(X)$. 
If in addition $(0,0)\in X$, then $W(X)_s$ is a weak* closed real operator system on $\hh$.

Let $\bor{\gg\times\ggh}$ be the Borel $\sigma$-algebra of the locally compact and second countable space $\Omega\equiv\gg\times\ggh$. For any state $\tau\in\state$, a {\em covariant phase space observable} with the \emph{fiducial state} $\tau$ is the following observable $\Co_\tau$ on $\gg\times\ggh$
\begin{equation*}
\Co_\tau(X) = \int_X W(x,\xi)\tau W(x,\xi)^*\de x \de\xi \qquad \forall X\in\bor{\gg\times\ggh}
\end{equation*}
(see \cite{Werner84} or \cite{CaCaDeToVa04} for the case of $\gg=\ggh=\R^n$, and \cite[Prop.~2.1, p.~166]{PSAQT82} and \cite[Theorem~3.4.2]{ATthesis} 
for the general form of covariant phase space observables). The integral in the definition of $\Co_\tau$ is understood in the weak*-sense, i.e., for all $S\in\trh$,
$$
\tr{S\Co_\tau(X)} = \int_X \tr{SW(x,\xi)\tau W(x,\xi)^*} \de x \de\xi \qquad \forall X\in\bor{\gg\times\ggh} \, .
$$

More generally, the map $(x,\xi)\mapsto \tr{SW(x,\xi)T W(x,\xi)^*}$ is continuous and integrable for all $S,T\in\trh$, with
\begin{equation}\label{eq:J*}
\int \tr{SW(x,\xi)T W(x,\xi)^*} \de x \de\xi = \tr{S}\tr{T}
\end{equation}
(see \cite[Lemma 3.1]{Werner84} for the case $\gg=\ggh=\R^n$, and \cite{KiSc20??} for the general case).

For any $T\in\trh$, the {\em inverse Weyl transform} of $T$ is the continuous function $\widehat{T}$ on $\gg\times\ggh$ given by
$$
\widehat{T} (x,\xi) = \tr{TW(x,\xi)} \, .
$$
The zero-level set of $T$ is the closed set
$$
Z(T)=\{(x,\xi)\in\gg\times \ggh \mid \widehat{T} (x,\xi) = 0\} \, .
$$
As usual, ${\rm supp}\, \widehat{T}$ stands for the support of $\widehat{T}$, that is, 
$$
{\rm supp}\, \widehat{T}=\overline{\{(x,\xi)\in\gg\times \ggh \mid \widehat{T} (x,\xi) \neq  0\}} \, .
$$

By injectivity of the inverse Weyl transform \cite{Werner84,KiSc20??}, we have $T=0$ if and only if $Z(T)=\gg\times\ggh$ or, equivalently, ${\rm supp}\, \widehat{T} = \emptyset$. Since $\widehat{T^*}(x,\xi) = \overline{\pair{\xi}{x}} \, \overline{\widehat{T}(-x,-\xi)}$, if $T\in\trhs$, then the sets $Z(T)$ and ${\rm supp}\, \widehat{T}$ are symmetric. Moreover, if $T$ is such that $\tr{T}\neq 0$, then $(0,0)\notin Z(T)$, since $\widehat{T}(0,0)=\tr{T}$.

\begin{proposition}\label{prop:covariant_nassela}
For any covariant phase space observable $\Co_\tau$ we have 
\begin{equation}\label{eqn:weyl_span}
\rr(\Co_\tau)= W({\rm supp}\, \widehat{\tau})_s
\end{equation}
and 
\begin{equation}\label{eqn:weyl_span_comp}
\begin{split}
\rr(\Co_\tau)^\perp & = \{S\in\trhs\mid {\rm supp}\, \widehat{\tau}\subseteq Z(S)\} \\
& = \{S\in\trhs\mid {\rm supp}\, \widehat{S}\subseteq Z(\tau)\} \, .
\end{split}
\end{equation}
\end{proposition}

\begin{proof}
Note  that  $S\in\rr(\Co_\tau)^\perp$ if and only if $\tr{SW(x,\xi)\tau W(x,\xi)^*}=0$ for all $(x,\xi)\in\gg\times \ggh$. In this case, using \eqref{eq:J**} and \eqref{eq:J*}, we obtain that the symplectic Fourier transform
$$
\int \overline{\pair{\zeta}{x}} \pair{\xi}{y} \tr{SW(y,\zeta)\tau W(y,\zeta)^*} \de y \de\zeta =  \widehat{S}(x,\xi)\overline{\widehat{\tau}(x,\xi)} \equiv 0
$$
so that $Z(\tau)^c\subseteq Z(S)$, but since $Z(S)$ is closed, this implies that ${\rm supp}\, \widehat{\tau}\subseteq Z(S)$. On the contrary, if $S\in\trhs$ is such that ${\rm supp}\, \widehat{\tau}\subseteq Z(S)$, then by the injectivity of the symplectic Fourier transform we have $\tr{SW(x,\xi)\tau W(x,\xi)^*}=0$ for all $(x,\xi)\in\gg\times \ggh$ so that $S\in \rr(\Co_\tau)^\perp$. This shows the first equality in \eqref{eqn:weyl_span_comp}. For the second, note that, if $T_1,T_2\in\trh$ are such that ${\rm supp}\, \widehat{T}_1 \subseteq Z(T_2)$, then
$$
{\rm supp}\, ( \widehat{T}_2)  = \overline{Z(T_2)^c} \subseteq \overline{({\rm supp}\, \widehat{T}_1)^c}\subseteq Z(T_1) \, .
$$
Therefore, ${\rm supp}\, \widehat{T}_1 \subseteq Z(T_2)$ $\Leftrightarrow$ ${\rm supp}\, \widehat{T}_2 \subseteq Z(T_1)$, hence the second equality in \eqref{eqn:weyl_span_comp} holds.

In order to complete the proof, we note that
\begin{align*}
W({\rm supp}\, \widehat{\tau})_s^\perp & = \{S\in\trhs\mid \tr{SW(x,\xi)} = 0 \ \forall (x,\xi) \in {\rm supp}\, \widehat{\tau}\} \\
& = \{S\in\trhs\mid {\rm supp}\, \widehat{\tau} \subseteq Z(S)\} \, .
\end{align*}
Comparing this with \eqref{eqn:weyl_span_comp}, \eqref{eqn:weyl_span} follows by the Bipolar Theorem.
\end{proof}

It follows from Prop.~\ref{prop:covariant_nassela} that, vaguely speaking, the larger the support of the inverse Weyl transform of the fiducial state is, the better the corresponding observable is from the state determination point of view. 
In particular, the extreme case ${\rm supp}\, \widehat{\tau} = \mathcal{G}\times \widehat{\mathcal{G}}$, if any, must yield an informationally complete observable. The next proposition shows that this is indeed a necessary and sufficient condition for informational completeness. The proof (taken from \cite{KiSc20??}) is a straightforward generalization of the corresponding result for the case $\mathcal{G}\times \widehat{\mathcal{G}}=\R^2$ proved in \cite{Kiukas2012}.

\begin{proposition}\label{prop:covariant_infocomplete}
The following conditions are equivalent:
\begin{enumerate}[(i)]
\item $\Co_\tau$ is informationally complete
\item ${\rm supp}\, \widehat{\tau} = \mathcal{G}\times \widehat{\mathcal{G}}$.
\end{enumerate}
\end{proposition}

\begin{proof}
If (ii) holds, then by Prop.~\ref{prop:covariant_nassela} and injectivity of the inverse Weyl transform we have $\rr(\Co_\tau)^\perp = \{0\}$, so that $\Co_\tau$ is informationally complete by Cor.~\ref{prop:well-known}. 

Suppose then that (ii) does not hold. In order to complete the proof, by \eqref{eqn:weyl_span_comp} we need to show that there exists a nonzero $S\in\trhs$ such that ${\rm supp}\, \widehat{\tau}\subseteq Z(S)$.
The set $U = ({\rm supp}\, \widehat{\tau})^c $ is nonempty, symmetric, open, and does not contain the identity $(0,0)$. Let $(x_0,\xi_0)\in U$. 
Then we can find a symmetric open neighbourhood $V$ of $(0,0)$ such that $V+V\subseteq  (U-(x_0,\xi_0))\cap (U+(x_0,\xi_0)) \equiv U_0$ and the measure $\vert V\vert $ of $V$ is finite.  The function $f = \chi_V *\chi_V$ (convolution in $\gg\times\ggh$) is then of positive type {\cite[Cor. 3.16]{CAHA95}}, so by Bochner's theorem {\cite[Theorem~4.18]{CAHA95}} there exists a positive bounded measure $\mu:\mathcal{B}(\mathcal{G}\times \widehat{\mathcal{G}})\to[0,\infty)$ such that $\widehat{\mu} = f$.

Let $S_0\in\trhs$, $S_0\geq 0$, be  nonzero and define 
$$
S_1 = \int W(x,\xi) S_0W(x,\xi)^*\, d\mu(x,\xi) 
$$ 
so that $S_1\geq 0$ and $\widehat{S_1}(x,\xi) = \widehat{\mu}(x,\xi) \widehat{S}_0(x,\xi) $. Now, $V \cap ((x,\xi) + V) = \emptyset$ for all $(x,\xi)\notin U_0$, since $(y,\zeta) \in V \cap ((x,\xi) + V)$ implies $(x,\xi) = (y,\zeta) + ((x,\xi) - (y,\zeta)) \in V-V \subset U_0$. Therefore, $\widehat{\mu}(x,\xi) = \vert V \cap ((x,\xi) + V)\vert = 0$ for all $(x,\xi)\notin U_0$ so that $\widehat{S_1}(x,\xi) =0$ for all $(x,\xi)\notin U_0$, but $S_1\neq 0$ since $\widehat{S_1}(0,0) = \vert V\vert\, \tr{S_0} \neq 0 $.  Finally, define $S_+ = W(x_0,\xi_0)S_1 + S_1W(x_0,\xi_0)^*$ and $S_- = W(x_0,\xi_0)S_1 - S_1W(x_0,\xi_0)^*$, so that at least one of $S_+$, $S_-$ is nonzero and
$$
\widehat{S_{\pm}}(x,\xi) = {\pair{\xi}{x_0}} \widehat{S_1}(x+x_0,\xi+\xi_0) \pm {\pair{\xi_0}{x_0-x}} \widehat{S_1}(x-x_0,\xi-\xi_0) \equiv 0
$$
for all $(x,\xi)\notin U$. Hence, ${\rm supp}\, \widehat{\tau} = U^c \subseteq Z(S_\pm)$ and, as $S_+ ,\, iS_-\in\trhs$, the proof is complete.
\end{proof}

%%%%%%%%%%%
\subsection{Finite dimension}
%%%%%%%%%%%

For any nonzero $d\in\nat$, we denote by $\Z_d$ the cyclic group with $d$ elements, and let $\omega = e^{2\pi i /d}$. Then, $\widehat{\Z_d} \equiv \Z_d$, the pairing of $x\in\Z_d$ and $\xi\in\widehat{\Z_d}$ being $\pair{\xi}{x} = \omega^{\xi x}$. 
Moreover, the Haar measures of $\Z_d$ and $\widehat{\Z_d}$ are just $d^{-1/2}\times$ the respective counting measures.
Let $\hh$ be a $d$-dimensional Hilbert space, and choose an orthonormal basis $\{\psi_j\}_{j\in\Z_d}$  of $\hi$. The Weyl map $W:\Z_d \times \Z_d \to \lh$ is then given by
$$
W(x,\xi)\psi_j = \omega^{\xi j}\psi_{j+x} \, .
$$
Since $\tr{W(x,\xi)^* W(y,\zeta)} = d \delta_{x,y} \delta_{\xi,\zeta}$ for all $(x,\xi),(y,\zeta)\in\Z_d\times\Z_d$, the set $\{d^{-1/2}W(x,\xi)\}_{(x,\xi)\in\Z_d\times\Z_d}$ is an orthonormal basis of the linear space $\lh$ endowed with the Hilbert-Schmidt inner product $\ip{A}{B} = \tr{A^* B}$.

Prop.~\ref{prop:covariant_nassela} now reduces to 
\begin{eqnarray}
\rr(\Co_\tau) &=& \spannoC{W(x,\xi) \mid (x,\xi)\in Z(\tau)^c}\cap\lhs \label{eqn:finite_weyl_span} \\
\rr(\Co_\tau)^\perp &=& \spannoC{W(x,\xi) \mid (x,\xi)\in Z(\tau)}\cap\lhs \label{eqn:finite_weyl_span_comp}
\end{eqnarray}
so that the real operator system is completely characterized by the zero set $Z(\tau)$. Therefore the essential question is the characterization of possible zero sets. This is done in the next proposition.

\begin{proposition}\label{prop:finite_zero_set}
For any state $\tau\in\state$, we have $-Z(\tau)=Z(\tau)$ and $(0,0)\notin Z(\tau)$. Conversely, if  $X\subset \Z_d \times \Z_d$ is such that $-X = X$ and $(0,0)\notin X$, then there exists a state $\tau$ such that $Z(\tau)=X$. 
\end{proposition}

\begin{proof}
We have already observed that $\tau\in\trhs$ and $\tr{\tau}>0$ implies the first part of the proposition.

For the second part, suppose first that  $X=\emptyset$. Then one may choose $\tau= \vert \psi\rangle\langle\psi\vert$ with $\psi = C \sum_{j=0}^{d-1} \alpha^j \psi_j$, where $0 < \vert\alpha\vert<1$ and $C>0$ is a normalization constant, which gives $\widehat{\tau} (x,\xi) = C^2 (\alpha\omega^\xi)^{-x} (1-|\alpha|^{2d})/(1-|\alpha|^2 \omega^\xi)$ and hence $Z(\tau)=\emptyset$. Suppose now that $X\neq \emptyset$.
For any $(x,\xi)\in X^c$ define the function 
$$
f_{(x,\xi)}(y,\zeta) = \frac{1}{d} (1 + \cos(2\pi (\zeta x- \xi y)/d)).
$$ 
By taking the symplectic Fourier transform we have for $(x,\xi)\neq (0,0)$
$$
\widehat{f}_{(x,\xi)} (z,\eta) = \frac{1}{d}\sum_{y,\zeta =0}^{d-1} e^{2\pi i (\eta y - \zeta z)/d} f_{(x,\xi)}(y,\zeta) =\left\{\begin{array}{ll} 1 & {\rm if }\, (z,\eta) = (0,0)\\ 1/2 & {\rm if }\, (z,\eta) = \pm(x,\xi) \\ 0 &{\rm otherwise}  \end{array}\right. \, ,
$$
and $\widehat{f}_{(0,0)} (0,0) = 2$ and $\widehat{f}_{(0,0)} (z,\eta) =0$ otherwise.
Now define $f = \sum_{(x,\xi)\in X^c} f_{(x,\xi)}$ which then satisfies 
$$
\widehat{f}(z,\eta) =0 \quad\Leftrightarrow\quad \widehat{f}_{(x,\xi)}(z,\eta) =0 \, \, \forall \, \, (x,\xi)\in X^c \quad\Leftrightarrow\quad (z,\eta)\in X \, .
$$ 
Let $\tau_0\in\ss$ be such that $\widehat{\tau_0} (x,\xi)\neq0$ for all $(x,\xi)\in\Z_d\times \Z_d$, and define
$$
\tau = \frac{1}{d(\# X^c +1)} \sum_{y,\zeta =0}^{d-1}  f(y,\zeta) W(y,\zeta) \tau_0 W(y,\zeta)^* \, .
$$
Then  $\tau\in\ss$ since $f\geq 0$ and $\sum_{y,\zeta =0}^{d-1}  f(y,\zeta) = d(\# X^c +1)$, and moreover $\widehat{\tau} (x,\xi) = (\# X^c +1)^{-1} \widehat{f} (x,\xi) \widehat{\tau_0} (x,\xi) \equiv 0$ if and only $\widehat{f} (x,\xi)=0$. 
That is, $Z(\tau) = X$.
\end{proof}

Since in the finite dimensional setting no topological considerations are needed, Prop.~\ref{prop:covariant_infocomplete} takes the following simple and well-known form.
\begin{proposition}\label{prop:eqZ0finite}
The  following conditions are equivalent:
\begin{enumerate}[(i)]
\item The observable $\Co_\tau$ is informationally complete.
\item $Z(\tau)=\emptyset$.
\end{enumerate}
\end{proposition}

The next result shows that for covariant phase space observables in dimensions $2$ and $3$ all of the notions of informational completeness are equivalent.  For $d=2$, indeed this is true for arbitrary observables (Example \ref{ex:2}); but the fact that in dimension $3$ all the notions of informational completeness are equivalent is specific to covariant phase space observables (compare with Example \ref{ex:3}).

\begin{proposition}
Let $d=2$ or $d=3$. Then the  following conditions are equivalent.
\begin{enumerate}[(i)]
\item The observable $\Co_\tau$ is $(\state^{\leq t},\state^{\leq p})$-informationally complete for some $t,p\in\nat$ such that $1\leq t\leq p\leq d$.
\item The observable $\Co_\tau$ is informationally complete.
\end{enumerate}
\end{proposition}

\begin{proof}
We already remarked that, in the case $d=2$, all the properties of $(S^{\leq t},S^{\leq p})$-informational completeness are equivalent by Example \ref{ex:2}. If $d=3$, then by Example \ref{ex:3} we still have to show that, for the observable $\Co_\tau$, $(\state^1,\state^1)$-informational completeness implies informational completeness. Now, the observable $\Co_\tau$ is $(\state^1,\state^1)$-informationally complete if and only if either $\rr(\Co_\tau)^\perp = \{ 0 \}$, in which case we are done, or  $\rr(\Co_\tau)^\perp  = \R T$ for some invertible $T\in\trhs$ with $\tr{T}=0$ by \cite[Prop.~5]{HeMaWo13}. In particular, $\dim \rr(\Co_\tau)^\perp =1$. In order to complete the proof we only need to show that this is not possible. By \eqref{eqn:finite_weyl_span_comp} and linear independence of the set $\{W(x,\xi)\}_{(x,\xi)\in\Z_d\times\Z_d}$, we have $\dim\rr(\Co_\tau)^\perp = \# Z(\tau)$. But since $Z(\tau)$ is symmetric, $(0,0)\notin Z(\tau)$ and the dimension $d=3$ is odd, $Z(\tau)$ must contain an even number of points, hence $\dim\rr(\Co_\tau)^\perp$ is even.
\end{proof}

As we have noted before, by increasing the size of the zero set the observable becomes less capable of performing state determination tasks. 
The next proposition shows that  already in the simplest case of an informationally \emph{incomplete} observable, namely, one having a zero set consisting of a single point, certain tasks become impossible. 

\begin{proposition}\label{prop:Z=1}
Suppose $d\geq 4$. 
The condition $\# Z(\tau) = 1$ can hold for some fiducial state $\tau$ only if $d$ is even.
If $\tau$ is a fiducial state with $\# Z(\tau) = 1$, then the observable $\Co_\tau$ is
\begin{enumerate}[(a)]
\item $(\state^{\leq t},\state)$-informationally complete for all $t < \frac{d}{2}$.
\item not $(\state^{\leq t},\state^{\leq t})$-informationally complete for any $t \geq \frac{d}{2}$.
\end{enumerate}
\end{proposition}

\begin{proof}
Let $Z(\tau) = \{(x,\xi)\}$ with $(x,\xi)\neq (0,0)$. 
Since $Z(\tau)$ is symmetric, we have $(x,\xi) = (-x,-\xi)$, and this implies that $d$ is even and $x=d/2$ or $\xi=d/2$. 
In particular, $\pair{\xi}{x} \in \{1,-1\}$. 
We fix a square root of $\pair{\xi}{x}$ and denote it by $\sigma$.
Then the operator $T\equiv\sigma W(x,\xi)$ is selfadjoint (by \eqref{eq:-x-xi}) and generates $\rr(\Co_\tau)^\perp$ (by \eqref{eqn:finite_weyl_span_comp}). 
Since $T^2 = \id$ and $\tr{T} = 0$, we have $\rank_+(T) = \rank_-(T) = d/2$. The rest of the claim then follows from Theorem~\ref{prop:structure}.
\end{proof}

For the next possible case, i.e., a zero set consisting of two points, we can give the following characterization, analogous to Prop. \ref{prop:Z=1}, in odd prime dimensions.

\begin{proposition}
Suppose $d$ is an odd prime number and $\tau$ is a fiducial state with $\# Z(\tau) = 2$.
The observable $\Co_\tau$ is
\begin{enumerate}[(a)]
\item $(\ss^{\leq t},\ss)$-informationally complete for all $t < \floor{\frac{d}{2}}$. 
\item not $(\ss^{\leq t},\ss^{\leq t})$-informationally complete  for any $t \geq \floor{\frac{d}{2}}$.
\end{enumerate}
\end{proposition}

\begin{proof}
We have $Z(\tau) = \{(x,\xi), (-x,-\xi)\}$ for some nonzero $(x,\xi)\in\Z_d\times\Z_d$. As $d$ is odd, $2$ has a multiplicative inverse in the ring $\Z_d$ and we can define the following projective representation $W'$ of $\Z_d\times\Z_d$ on $\hh$
$$
W'(y,\zeta) = \omega^{2^{-1}\zeta y} W(y,\zeta) \, .
$$
Note that, as $W'(x,\xi)^* = W'(-x,-\xi)$, the operators 
$$
T_\alpha \equiv \alpha W'(x,\xi) + \overline{\alpha} W'(-x,-\xi)
$$ are selfadjoint for all $\alpha\in\C$, and $\rr(\Co_\tau)^\perp = \{T_\alpha\mid \alpha\in\C\}$ by \eqref{eqn:finite_weyl_span_comp}. For all $(y,\zeta)\in\Z_d\times\Z_d$, the map
$$
W'_{(y,\zeta)} : \Z_d \to \lh \, , \qquad W'_{(y,\zeta)} (t) = W'(ty,t\zeta)
$$
is actually a unitary representation of $\Z_d$, which is equivalent to the representation $V$ \cite{CaHeTo20??}. So, there is a Hilbert basis $\{\fii_\eta\}_{\eta\in\Z_d}$ such that $W'(ty,t\zeta) \fii_\eta = \omega^{\eta t} \fii_\eta$ for all $t\in\Z_d$. In particular, for $\alpha\neq 0$ the eigenvalues of $T_\alpha$ are $\{r\cos(2\pi t / d + \theta)\}_{t=0}^{d-1}$ for fixed $r\in\R_+$ and $\theta\in [0,2\pi)$. We thus see that the following three possibilities occur:
\begin{enumerate}
\item $\rank_+ (T_\alpha) = \rank_- (T_\alpha) = (d-1)/2$;
\item $\rank_+ (T_\alpha) = (d+1)/2$, \quad $\rank_- (T_\alpha) = (d-1)/2$;
\item $\rank_+ (T_\alpha) = (d-1)/2$, \quad $\rank_- (T_\alpha) = (d+1)/2$.
\end{enumerate}
In all three cases we see that $\rank_\downarrow(T_\alpha) \geq (d-1)/2$, and thus $\Co_\tau$ is $(\ss^{\leq t},\ss)$-informationally complete for $t = (d-1)/2 - 1$ by Theorem~\ref{prop:structure}. 
Moreover, choosing $T_\alpha$ as in item (1), by the same Theorem we see that $\Co_\tau$ is not $(\ss^{\leq t},\ss^{\leq t})$-informationally complete for $t = (d-1)/2$.
\end{proof}

We remark that if the dimension $d$ is not an odd prime, then in the case $\#Z(\tau)= 2$ the observable $\Co_\tau$ need not be $(\ss^1,\ss^1)$-informationally complete. 
Indeed, fix $d=4$ and let $Z(\tau) =\{ (0,1),(0,3)\}$. 
Then 
$$
\rr(\Co_\tau)^\perp = \{ \alpha W(0,1) + \beta W(0,3) \mid \alpha, \beta\in\C \}\cap\lhs \, ,
$$
and the elements of $\rr(\Co_\tau)^\perp$ are thus
$$
A(\alpha) = \left(\begin{array}{cccc} 
\alpha  + \overline{\alpha}  & 0 & 0 & 0 \\
0 & i (\alpha -  \overline{\alpha}  )& 0 & 0\\
0 & 0 & -(\alpha + \overline{\alpha} )& 0 \\
0 & 0 & 0 & -i (\alpha  -  \overline{\alpha})
\end{array}\right)\qquad \alpha \in\C \, .
$$
Now, for instance any  $\alpha\in\R$, $\alpha \neq 0$, will give $\rank A(\alpha) =  2$ which  implies that $\Co_\tau$ is not $(\ss^1,\ss^1)$-informationally complete by Cor.~\ref{prop:S1S1}.

As a final result concerning the finite dimensional phase space, we show that there is an upper bound on the size of the zero set after which the observable fails to be even $(\state^1,\state^1)$-informationally complete.

\begin{proposition}
Let $d\geq 4$ and denote by $\alpha$ the number of $1$'s in the binary expansion of $d-1$. If
\begin{enumerate}[(a)]
\item $\# Z(\tau) \geq (d-2)^2 + 2\alpha - 1$, or 
\item  $\# Z(\tau) \geq (d-2)^2  + 2\alpha - 3$, $d$ is odd and $ \alpha=3\,{\rm mod}\, 4$, or
\item $\# Z(\tau) \geq (d-2)^2  + 2\alpha - 2$, $d$ is odd and $\alpha=2\,{\rm mod}\, 4$,
\end{enumerate}
then $\Co_\tau$ is not $(\state^1,\state^1)$-informationally complete.
\end{proposition}

\begin{proof}
If $\# Z(\tau)$ is as in the statement, then by \eqref{eqn:finite_weyl_span} we have
$$
\dim \rr(\Co_\tau)\leq \left\{\begin{array}{cl} 4d-3 -2\alpha & \mbox{ in case }(a) \\
4d-1 -2\alpha & \mbox{ in case }(b)\\
4d-2 -2\alpha & \mbox{ in case }(c) 
\end{array}\right. \, ,
$$
and $\Co_\tau$ is not $(\state^1,\state^1)$-informationally complete by \cite[Theorem~6]{HeMaWo13}.
\end{proof}

%%%%%%%%%%%%%%%%
\subsection{Infinite dimension}
%%%%%%%%%%%%%%%%

Let $\gg = \R^n$, with dual group $\widehat{\R^n}\equiv \R^n$, pairing $\pair{\xi}{x} = e^{i \xi x}$ and Haar measures on $\gg$ and $\ggh$ coinciding with $(2\pi)^{-n/2}\times$ the Lebesgue measure.. Then, the Weyl map acts on the Hilbert space $\hh=L^2(\R^n)$ and is given by $W(x,\xi)=e^{-ix\cdot P}e^{i\xi\cdot Q}$ for all $x,\xi\in\R^n$, where $Q$ and $P$ are the usual $n$-dimensional position and momentum operators. In this case, by \eqref{eqn:weyl_span} it is the support of $\widehat{\tau}$ that is relevant for $\rr(\Co_\tau)$. 
The characterization of the possible supports is a difficult task, but for the possible zero sets $Z(\tau)$ this can be done. 

\begin{proposition}
For any state $\tau\in\ss$, $Z(\tau)$ is a closed set such that  $-Z(\tau)=Z(\tau)$ and $(0,0)\notin Z(\tau)$. Conversely, if $X\subseteq \R^{2n}$ is a closed set such that $-X = X$ and $(0,0)\notin X$, then there exists a state $\tau\in\ss$ such that $Z(\tau)=X$. 
\end{proposition}

\begin{proof}
We already remarked that $Z(\tau)$ is closed and symmetric, and $\tr{\tau} > 0$ implies $(0,0)\notin Z(\tau)$. In order to prove the converse statement, suppose first that  $X=\emptyset$. Choosing $\psi(x) = \pi^{-n/4} e^{-\|x\|^2/2}$ and defining $\tau = \kb{\psi}{\psi}$, it is easy to check that $Z(\tau)=\emptyset$. Now let $X\subseteq\R^{2n}$ be a closed nonempty set such that $-X = X$ and $(0,0)\notin X$. By \cite{RubinSellke}, there exists a probability measure $\mu:\bor{\R^{2n}}  \to[0,1]$ such that 
$$
\widehat{\mu} (x,\xi) =  0 \quad \mbox{if and only if} \quad (x,\xi)\in X \, .
$$ 
Let $\tau_0\in\state$ be such that $\widehat{\tau_0} (x,\xi)\neq 0$ for all $(x,\xi)\in\R^{2n}$, and define
$$
\tau = \int W(x,\xi) \tau_0 W(x,\xi)^* \de\mu(x,\xi) \, .
$$
Then $\tau$ is positive, nonzero, and satisfies $\widehat{\tau} (x,\xi) = \widehat{\mu} (x,\xi) \widehat{\tau_0} (x,\xi) \equiv 0$ if and only if $(x,\xi)\in X$, i.e., $Z(\tau)=X$. 
\end{proof}

In the finite dimensional setting we saw that the cardinality $\#Z(\tau)$ provides a natural way to characterize the state determination properties of the corresponding observables. In particular, observables having a small zero set are able to perform more tasks than those having larger zero sets. We see from the next proposition that in the infinite dimensional case a similar natural property is the compactness of the set.

\begin{proposition}\label{prop:infinite_dimension}
\begin{enumerate}[(a)]
\item ${\rm supp}\, \widehat{\tau}=\R^{2n}$ if and only if $\Co_\tau$ is informationally complete.
\item If $Z(\tau)$ is compact,  then $\Co_\tau$ is $(\state^{\rm fin},\state^{\rm fin})$-informationally complete.
\item If ${\rm supp}\, \widehat{\tau}$ is compact,  then $\Co_\tau$ is not $(\state^1,\state^1)$-informationally complete.
\item If neither $Z(\tau)$ nor ${\rm supp}\, \widehat{\tau}$ is compact, then $\Co_\tau$ may or may not be $(\state^1,\state^1)$-informationally complete.
\end{enumerate}
\end{proposition}

\begin{proof}
Part (a) is just a restatement of Prop.~\ref{prop:covariant_infocomplete}. 

Suppose then that $Z(\tau)$ is compact. By \cite[Theorem~2.2]{NarRat}, any $S\in\trhs$ such that $\widehat{S}$ is compactly supported, is necessarily of infinite rank. In particular, by \eqref{eqn:weyl_span_comp} and compactness of $Z(\tau)$ every $S\in\rr(\Co_\tau)^\perp$ has $\rank_\uparrow S = \infty$. 
Item (b) then follows from Theorem~\ref{prop:structure-fin}a.

Assume next that ${\rm supp}\, \widehat{\tau}$ is compact and let $R>0$ be such that ${\rm supp}\, \widehat{\tau}\subseteq B^{2n}_R = \{(x,\xi)\in\R^{2n}\mid \|x\|^2+\|\xi\|^2 < R^2\}$. Let $x_0\in\R^n$ be such that $\|x_0\| = 2R$, and denote $X= x_0 + B^n_{R/4} = \{x\in\R^n\mid \|x-x_0\|^2 < (R/4)^2\}$. Now define the unit vectors
$$
\psi_\pm = C (\chi_{-X} \pm \chi_{X})
$$
where $\chi_X$ denotes the characteristic function of the set $X$ and $C$ is the normalization constant. If $(x,\xi)\in {\rm supp}\, \widehat{\tau}$, then $\|x\|  < R$ and we have
\begin{align*}
\ip{\chi_{-X}}{W(x,\xi) \chi_{X}} & = \ip{\chi_{X}}{W(x,\xi) \chi_{-X}}\\
& = e^{-ix\cdot\xi}\,\int e^{i\xi\cdot z}  \chi_{-X}(z) \chi_{X} (z-x)\de z = 0 \, .
\end{align*}
Therefore, if $\rho_{\pm} = \kb{\psi_{\pm}}{\psi_{\pm}}$, then for $(x,\xi)\in {\rm supp}\, \widehat{\tau}$ we have
$$
\widehat{\rho}_{\pm} (x,\xi) = \vert C\vert^2  \left( \ip{\chi_{X}}{W(x,\xi) \chi_{X}} + \ip{\chi_{-X}}{W(x,\xi) \chi_{-X}} \right) \, ,
$$
hence
$$
\widehat{(\rho_+ - \rho_-)} (x,\xi) = 0 \, .
$$
By \eqref{eqn:weyl_span_comp}, $\rho_+ - \rho_- \in \rr(\Co_\tau)^\perp$, hence (c) holds by Prop.~\ref{prop:prem}.

For part (d) we will  give two examples in the case $n=1$. For the first one, let $R>0$ and choose $Z(\tau) = \left( [-6R,-R]\cup[R,6R]  \right)\times \R$. Then a similar argument as before shows that $\rho_+ - \rho_- \in \rr(\Co_\tau)^\perp$, hence $\Co_\tau$ is not $(\state^1,\state^1)$-informationally complete by Prop.~\ref{prop:prem}. For the second example, we refer to \cite[Prop.~9]{Kiukas2012} where the authors constructed a state $\tau$ such that $Z(\tau)$ is nowhere dense but of infinite Lebesgue measure. In other words, $\Co_\tau$ is informationally complete but neither $Z(\tau)$ nor ${\rm supp}\, \widehat{\tau}$ is compact.
\end{proof}

%%%%%%%%%%%
\subsection{An application: noisy measurements}
%%%%%%%%%%%

In any realistic measurement one needs to take into account the effect of noise originating from various imperfections in the measurement setup. This typically results in a smearing of the measurement outcome distribution which appears in the form of a convolution: if $\varrho^{\Co_{\tau_0}}$ is the probability distribution corresponding to the ideal measurement of $\Co_{\tau_0}$, the actually measured distribution is $\mu*\rho^{\Co_{\tau_0}}$ for some probability measure $\mu$ modelling the noise. The convolution does not affect the covariance properties of the observable and hence the general structure of the observable remains the same. That is, the actually measured observable is a covariant phase space observable $\Co_\tau$ with the smeared fiducial state 
$$
\tau =  \int W(x,\xi) \tau_0 W(x,\xi)^*\, \de\mu (x,\xi) \, .
$$
The inverse Weyl transform of $\tau$ now gives $\widehat{\tau}(x,\xi) = \widehat{\mu}(x,\xi) \widehat{\tau_0}(x,\xi)$. In particular, we have   $Z(\tau) = Z(\mu) \cup Z(\tau_0)$ where we have defined analogously $Z(\mu) = \{ (x,\xi)\in\mathcal{G}\times \widehat{\mathcal{G}} \mid \widehat{\mu}(x,\xi) =0\}$. 

Consider next the special case where $Z(\tau_0)=\emptyset$ so that $\Co_{\tau_0}$ is informationally complete. For instance, one may think of the measurement of the Husimi $Q$-function of a state, in which case $\hi=L^2(\R)$ and $\tau_0=\vert \psi_0\rangle\langle \psi_0 \vert$ is the vacuum, i.e., the ground state of the harmonic oscillator $\psi_0(x)  = \pi^{-1/4} e^{-x^2/2}$. Now the overall observable's ability to perform any state determination task is completely determined by the support of $\widehat{\mu}$. In the specific example with the $Q$-function we immediately see, e.g.,  that any Gaussian noise has no effect on the success of the task at hand. However,  from Prop.~\ref{prop:infinite_dimension} we know that any $\mu$ with ${\rm supp}\, \widehat{\mu} \neq \R^2$ but with compact $Z(\mu)$ results in an observable which is not informationally complete but still allows one to determine any finite rank state under the premise that the rank is bounded by some arbitarily high finite number $p$. Finally, if ${\rm supp}\, \widehat{\mu}$ is compact, then even the simplest task of determining pure states among pure states fails. 

%%%%%%%%%%%%%%%%%%%%%%%%%%%%%%%%%%%%%%%
\section*{Acknowledgements}
%%%%%%%%%%%%%%%%%%%%%%%%%%%%%%%%%%%%%%%%

T.H.~acknowledges financial support from the Academy of Finland (grant no. 138135). J.S.~and A.T.~acknowledge financial support of the Italian Ministry of Education, University and Research (FIRB project RBFR10COAQ).

%%%%%%%%%
%%%%%%%%%
\end{document}